\documentclass[11pt]{article}
\usepackage{algorithm}
\usepackage{algpseudocode}
\usepackage{amssymb, amsmath, amsthm}

\usepackage{amsfonts,mathtools,mathrsfs,mathtools,color}
\usepackage{bbm}
\usepackage[colorlinks]{hyperref}
\usepackage{graphicx}
\usepackage{microtype}
\usepackage{caption}
\usepackage{subcaption}
\usepackage{xcolor}

\newtheorem{theorem}{Theorem}
\newtheorem*{theorem*}{Theorem}
\newtheorem{lemma}{Lemma}
\newtheorem{corollary}{Corollary}

\newtheorem{assumption}{Assumption}

\newcommand{\R}{\mathbb{R}} 

\newcommand{\N}{\mathcal{N}}
\newcommand{\bN}{\mathbb{N}}

\newcommand{\E}{\mathbb{E}}

\renewcommand{\part}[2]{\frac{\partial #1}{\partial #2}}

\newcommand{\KL}{\mathsf{KL}}
\newcommand{\FI}{\mathsf{FI}}
\newcommand{\Tr}{\mathsf{Tr}}

\newcommand{\Var}{\mathsf{Var}}

\newcommand{\op}{\mathsf{op}}
\newcommand{\HS}{\mathsf{HS}}
\newcommand{\err}{\varepsilon}
\renewcommand{\d}{\mathrm{d}}
\newcommand{\dt}{\mathrm{d}t}
\newcommand{\dx}{\mathrm{d}x}
\newcommand{\dy}{\mathrm{d}y}
\newcommand{\ddt}{\frac{\d}{\dt}}
\newcommand{\Q}{\mathsf{Q}}
\newcommand{\CLSI}{C_\mathsf{LSI}}

\newcommand{\deq}{\coloneqq}
\newcommand{\one}{\mathbbm{1}}
\newcommand{\RGO}{\mathsf{R}}
\newcommand{\hRGO}{\widehat{\mathsf{R}}}
\newcommand{\mmid}{\,\|\,}

\DeclareMathOperator{\poly}{poly}
\DeclareMathOperator{\polylog}{polylog}
\DeclareMathOperator{\prox}{prox}

\DeclareMathOperator{\TV}{\mathsf{TV}}
\renewcommand{\Pr}{\mathbb P}
\DeclarePairedDelimiter{\norm}{\lVert}{\rVert}
\newcommand{\msf}[1]{\mathsf{#1}}
\newcommand{\Ren}{\mathcal{R}}
\newcommand{\Alg}{\mathtt{Alg}}

\title{Complexity of Non-Log-Concave Sampling in Fisher Information}
\author{Sinho Chewi\thanks{Yale University, Department of Statistics and Data Science. Email: \texttt{sinho.chewi@yale.edu}.}
\and Andre Wibisono\thanks{Yale University, Department of Computer Science. Email: \texttt{andre.wibisono@yale.edu}. This work was supported by NSF awards CCF--2403391 and CAREER CCF--2443097.}}
\date{\today}

\begin{document}

\maketitle

\begin{abstract}
    We study the query complexity of obtaining a relative Fisher information guarantee for sampling from a log-smooth non-log-concave distribution; this is a sampling analog of finding an approximate stationary point in optimization.
    Our algorithm is based on the proximal sampler, which is an implicit discretization of the Langevin diffusion, and requires an implementation of the backward step known as the restricted Gaussian oracle (RGO).
    We show that by leveraging the recent results for log-concave sampling with high-accuracy guarantees in R\'enyi divergence, we can obtain an approximate RGO implementation that---when used with the proximal sampler---yields a complexity guarantee in relative Fisher information that inherits the same dimension dependence as log-concave sampling, and improves upon prior work for non-log-concave sampling.
    We also show a converse reduction that any improvement in the dimension dependence in relative Fisher information for non-log-concave sampling will yield an improved dimension dependence for high-accuracy log-concave sampling.
\end{abstract}

\section{Introduction}

\subsection{Motivation and background}\label{ssec:motivation}

What is the complexity of sampling from a smooth distribution $\pi \propto \exp(-f)$ over $\R^d$?
When $\pi$ is log-concave or satisfies an isoperimetric inequality, then this question is well-studied, with many results providing iteration complexity guarantees which scale polynomially in the dimension and the isoperimetric constants; see~\cite{Chewi26Book} for an introduction to some of these developments.

In this work, we are interested in the opposite scenario, in which $\pi$ is allowed to be highly non-log-concave, with the only assumption being that $f$ is $L$-smooth, in the sense that $-LI \preceq \nabla^2 f \preceq LI$.
In this case, simple worst-case examples demonstrate that generating a sample close in total variation distance to $\pi$ can require exponentially many queries to $f$ (e.g., with respect to the ambient dimension $d$; see~\cite{HeZha25NonLogConcave}.).
One way to overcome this difficulty is to impose additional problem structure, such as assuming that $\pi$ is a mixture of strongly log-concave distributions~\cite{LeeRisGe18Multimodal} or to assume that an annealing path has small ``action''~\cite{GuoTaoChe25Annealed}, which requires designing tailor-made algorithms for each problem class under consideration.
Another approach is to weaken the notion of progress by seeking an approximately ``stationary'' solution.

The latter idea has been especially fruitful in the sister field of optimization~\cite{Nes12GradSmall}, namely to minimize (or find a stationary point of) an objective function $f \colon \R^d \to \R$, in which the corresponding question becomes: How many queries to $\nabla f$ are necessary to reach a point $x$ with $\norm{\nabla f(x)} \le \varepsilon$?
In the case where $f$ is $L$-smooth but potentially non-convex, gradient descent finds such an approximate stationary point in $O(1/\varepsilon^2)$ iterations, and this rate is optimal in high dimension~\cite{Car+20StatPt}.
The optimal complexity in low dimension is still open, and it turns out that the results of this paper take on an interesting interpretation in light of this story; see Section~\ref{ssec:opt_analogy} below.
Although finding a stationary point is still a far cry from finding a global optimum, it is often the starting point for deeper investigations, such as the development of algorithms which escape ``undesirable'' stationary points (e.g., strict saddles).
More broadly, this theory has the advantage of having great generality, as it is one of the few frameworks available to quantitatively study general non-convex optimization problems.

Developing an analogous theory for the task of sampling is more subtle, as one must first decide on a suitable definition of an ``approximate stationary point''.
Such a definition was proposed in~\cite{balasubramanian22a}, and it is based on the celebrated JKO theorem~\cite{JKO} which interprets the Langevin diffusion as a gradient flow, as explained next.
The Langevin diffusion is the stochastic process for $X_t \in \R^d$ governed by the stochastic differential equation (SDE)
\begin{align}\label{eq:langevin}
    \d X_t = -\nabla f(X_t)\,\dt + \sqrt 2\,\d B_t\,,
\end{align}
where $B_t$ is the standard Brownian motion in $\R^d$.
Langevin diffusion is the basis of some of the most popular sampling algorithms, such as the unadjusted Langevin algorithm (ULA) and the Metropolis-adjusted Langevin algorithm (MALA), as well as the proximal sampler studied in this paper, all of which can be viewed as time discretization of the Langevin diffusion.
If $\rho_t$ denotes the law of $X_t$ evolving along Langevin diffusion as in~\eqref{eq:langevin}, then~\cite{JKO} interprets the curve of probability measures $(\rho_t)_{t\ge 0}$ as a gradient flow of the relative entropy or KL divergence functional $\KL(\cdot \mmid \pi)$ over the space of probability measures equipped with the Wasserstein $W_2$ metric.
See~\cite{Wib18SamplingOpt} for further discussion on this perspective.

The gradient flow perspective also brings with it a set of calculation rules, usually termed ``Otto calculus'', such as the following identity for the dissipation of the KL divergence along Langevin diffusion~\eqref{eq:langevin}, classically known as de Bruijn's identity:
\begin{align}\label{eq:debrujin}
    \ddt\, \KL(\rho_t \mmid \pi) = -\FI(\rho_t\mmid \pi)\,,
\end{align}
where $\FI(\rho_t\mmid \pi)$ is the relative Fisher information of $\rho_t$ relative to $\pi$:
\begin{align*}
     \FI(\rho_t \mmid \pi) \deq \int \Bigl\| \nabla \log \frac{\d\rho_t}{\d\pi}\Bigr\|^2\,\d\rho_t\,.
\end{align*}
In the context of Otto calculus, the relative Fisher information takes on a new interpretation: it is the squared norm of the \emph{Wasserstein gradient} of $\KL(\cdot\mmid \pi)$ at $\rho_t$.
This motivated~\cite{balasubramanian22a} to define an $\varepsilon$-stationary point for sampling to be a probability measure $\rho$ for which $\sqrt{\FI(\rho \mmid \pi)} \le \varepsilon$, and the complexity question for non-log-concave sampling becomes:
\begin{center}
    \emph{What is the complexity of outputting a sample from a probability measure $\rho$ such that $\sqrt{\FI(\rho\mmid \pi)} \le \varepsilon$?}
\end{center}

In that same paper,~\cite{balasubramanian22a} showed that the unadjusted Langevin algorithm, i.e., the Euler--Maruyama discretization of~\eqref{eq:langevin}, can achieve this goal in $O(1/\varepsilon^4)$ steps.\footnote{Here, the big-O hides dependence on the dimension $d$, the initial KL divergence, and the smoothness $L$.}
Subsequently, the work of~\cite{Che+23FisherLower} established a query lower bound of $\Omega(1/\varepsilon^{2-o(1)})$ in high dimension, and~\cite{ZhoSug25FI} developed an improved algorithm using parallel computation which runs for $\widetilde O(1/\varepsilon^2)$ parallel rounds and makes $\widetilde O(1/\varepsilon^6)$ total queries.

Approximate stationarity in sampling admits a compelling interpretation as measuring closeness to \emph{metastability}. Indeed, it is well-established that in the presence of free energy barriers, e.g., produced by well-separated modes of the target distribution $\pi$, local Markov chains often get stuck at certain distributions which persist for exponential stretches of time, called metastable states.
Although the concept of a metastable state is only loosely defined, in light of~\eqref{eq:debrujin}, we can view the relative Fisher information as a quantitative measure of metastability, since it implies that a measure $\rho$ with small relative Fisher information will experience slow decrease of the KL divergence under the Langevin diffusion.
However, we also note that even when $\pi$ is non-log-concave and has multiple modes, the KL divergence $\rho \mapsto \KL(\rho \mmid \pi)$ does not have spurious local minima; the only stationary point is the global minimizer $\rho = \pi$.

These interpretations have been further developed in the setting where $\pi$ is a mixture, $\pi = \sum_{k=1}^K w_k \pi_k$ for some weights $w_1,\dots,w_K > 0$ with $\sum_{k=1}^K w_k = 1$, and some component distributions $\pi_1,\dots,\pi_K$. 
For example, when $\pi$ is a mixture of Gaussians,~\cite{Cheng+23CondMixing} showed that a FI guarantee $\FI(\rho\mmid \pi) \le \varepsilon^2$ implies mixing for each of the components which have non-trivial mass under $\rho$.
More recently,~\cite{Nil26Reweighted} showed that if each $\pi_k$ satisfies a functional inequality (e.g., log-Sobolev), then $\FI(\rho\mmid \pi) \le \varepsilon^2$ implies that $\rho$ is close in KL divergence to a \emph{reweighted} mixture $\sum_{k=1}^K w_k' \pi_k$ for some (explicit) weights $w_1',\dots,w_K'$.
Generally, these results indicate that the complexity of stationarity in sampling captures the complexity of \emph{local} mixing, which is a useful complement to worst-case \emph{global} mixing results for non-log-concave sampling.

There remains, however, a gap between the upper and lower bounds in the sequential setting, leading to the question of determining an optimal algorithm for local mixing.
In particular, analogously to the picture for optimization, is an optimal algorithm given by some discretization of the Langevin diffusion---the sampling equivalent of the gradient flow---and if so, which one?

In this work, we close this gap and answer this question affirmatively.
We develop an algorithm, based on the proximal sampler~\cite{LST21structured, CCSW22improved}, which achieves the optimal complexity of $O(1/\varepsilon^2)$ up to logarithmic factors.
Moreover, as we discuss below, the dimension dependence of our result is the best possible without further advances in the complexity of log-concave sampling.

\subsection{Contributions}

Our results are based on the \textit{proximal sampler} algorithm, which is reviewed in Section~\ref{ssec:prox_sampler}. 
Briefly, the proximal sampler first augments the target distribution $\pi^X(x) \propto \exp(-f(x))$ into a joint density over two variables, $\pi^{XY}(x,y) \propto \exp(-f(x)-\frac{1}{2h}\,\|y-x\|^2)$, and then applies Gibbs sampling. Concretely, this consists of alternating between the following two steps: given $X$, sample $Y \sim \pi^{Y\mid X} = \N(X, hI_d)$; given $Y$, sample $X \sim \pi^{X\mid Y}$, where $\pi^{X\mid Y}$ is known as the restricted Gaussian oracle (RGO).
When $\pi^X$ is log-smooth and satisfies an isoperimetric inequality, such as a log-Sobolev inequality,~\cite{CCSW22improved} shows a high-accuracy guarantee for the proximal sampler in R\'enyi divergence.
When $\pi^X$ is log-smooth and strongly log-concave,~\cite{wibisono2025mixing} shows a high-accuracy guarantee for the proximal sampler in relative Fisher information.
In this paper, we work only under the assumption that $\pi^X$ is log-smooth, without assuming log-concavity or isoperimetry.

Our first contribution is to show that under log-smoothness alone, the relative Fisher information decays in an average sense along the proximal sampler.
Namely, we establish that
\begin{align}\label{eq:ideal_FI_bd}
    \FI(\bar\rho_K^X \mmid \pi^X) \le \frac{1}{K} \sum_{k=1}^K \FI(\rho_k^X \mmid \pi^X) \le \frac{\KL_0}{h\,(1-\frac{Lh}{2})\,K}\,,
\end{align}
where $\rho_k^X$ denotes the law of the $k$-th iterate of the proximal sampler algorithm, $\bar\rho_K^X \deq \frac{1}{K} \sum_{k=1}^K \rho_k^X$ is the averaged measure, and $\KL_0 \deq \KL(\rho_0^X \mmid \pi^X)$; see Theorem~\ref{Thm:FIExactRGO} below.
We note such a result is expected in light of the de Bruijn identity~\eqref{eq:debrujin}, which implies that along the Langevin diffusion,
\begin{align*}
    \frac{1}{T} \int_0^T \FI(\rho_t\mmid \pi^X)\,\dt \le \frac{\KL_0}{T}\,.
\end{align*}
Establishing the corresponding guarantee~\eqref{eq:ideal_FI_bd} for the proximal sampler is more involved since it requires tracking the relative Fisher information along simultaneous evolution of both arguments along the heat flow and its time reversal, and here we leverage the recent results of~\cite{wibisono2025mixing}.

However,~\eqref{eq:ideal_FI_bd} is not yet an algorithmic result, as it assumes that we can implement the RGO exactly\@.
Our next contributions address this point.
First, by choosing $h \asymp 1/(dL)$ and by implementing the RGO exactly via rejection sampling following~\cite{CCSW22improved}, the result~\eqref{eq:ideal_FI_bd} implies an $O(\frac{Ld \, \KL_0}{\err^2})$ query complexity guarantee for the proximal sampler to achieve $\FI(\bar\rho_K^X \mmid \pi^X) \le \err^2$; see Corollary~\ref{Cor:FIExactRGOComplexity}.
We note this already improves upon the $O(\frac{L^2 d \, \KL_0}{\err^4})$ complexity obtained in~\cite{balasubramanian22a} via ULA.
Our next results leverage the recent approximate RGO implementations~\cite{FanYuaChe23ImprovedProx, AltChe24Warm, Chen+26HighAccDiffusion} to further improve the dimension dependence.

By choosing $h = 1/(2L)$, one can see that the RGO $\pi^{X \mid Y=y}$ is $L$-strongly log-concave and $3L$-log-smooth, which means that it is a log-concave distribution with condition number at most $3$.
Thus, it is natural to use existing log-concave samplers to implement the RGO\@.
There are two main obstacles to this approach:
\begin{enumerate}
    \item Even if we have a log-concave sampler which implements the RGO accurately in the sense of relative Fisher information, i.e., the sampler yields an approximate RGO $\hat\pi^{X\mid Y=y}$ such that $\FI(\hat\pi^{X\mid Y=y} \mmid \pi^{X\mid Y=y}) \le \delta^2$, it is non-trivial to track the propagation of this error along the proximal sampler, as the relative Fisher information lacks natural properties such as a data-processing inequality.
    See Section~\ref{ssec:pf_outline} for discussion.
    \item Prior works on log-concave sampling usually do not yield relative Fisher information guarantees in the first place.
\end{enumerate}
We address these issues by showing that, given a \emph{high-accuracy} log-concave sampler with R\'enyi divergence guarantees---namely, an algorithm which yields a sample from any $O(1)$-conditioned log-concave density up to error $\varepsilon^2$ in $\Ren_3(\cdot \mmid \pi^X)$ (R\'enyi divergence of order $3$) in $\widetilde O(d^p \polylog(1/\varepsilon))$ iterations---perturbing the output of the sampler by a Gaussian with a small but well-chosen variance yields an RGO implementation whose error is well-controlled along the proximal sampler.
This leads to an algorithm which obtains $\FI(\cdot\mmid \pi^X) \le \varepsilon^2$ using
\begin{align}\label{eq:informal_result}
    \frac{L\,\KL_0}{\varepsilon^2} \times d^p \times \polylog\bigl(\Delta_0, L/\varepsilon^2, d\bigr)\qquad\text{iterations}\,,
\end{align}
where $\Delta_0 \deq \Ren_2(\rho_0^X \mmid \pi^X)$ arises as an artifact of our analysis; see Theorem~\ref{thm:main}.
In particular, although the optimal exponent $p$ is still unknown, current state-of-the-art log-concave samplers achieve this high-accuracy guarantee with exponent $p=\frac{1}{2}$~\cite{FanYuaChe23ImprovedProx, AltChe24Warm, Chen+26HighAccDiffusion}, which we can easily plug into the result above.

Note that, up to logarithmic factors, the dependence on $L/\varepsilon^2$ is \emph{optimal} in high dimension, in light of the lower bound of~\cite{Che+23FisherLower}.
If we compare against the analogous problem of finding approximate stationary points in non-convex optimization, there the optimal complexity in high dimension is $L\Omega_0/\varepsilon^2$, where $\Omega_0 \deq f(x_0) - \inf f$ is the initial objective gap~\cite{Car+20StatPt}.
Thus, our complexity bound is entirely similar to the one in optimization, except that we incur an extra dimensional factor $d^p$, which is the complexity of high-accuracy log-concave sampling.
Thanks to our general reduction, any improvement to the complexity of the latter task would yield an improved Fisher information complexity result for non-log-concave sampling.

However, one could still ask whether there is a different approach which avoids the reduction to log-concave sampling altogether and could therefore potentially lead to improved dimension dependence.
Our final result says, in essence, that the answer is no.
More precisely, we show that any Fisher information complexity bound of the form~\eqref{eq:informal_result} implies the existence of a log-concave sampler with complexity $\widetilde O(d^p \polylog(1/\varepsilon))$, see Theorem~\ref{thm:converse}.
This closes the loop and shows that the optimal dimension dependence for the Fisher information complexity of non-log-concave sampling, and for high-accuracy log-concave sampling, coincide.

\subsection{An analogy with optimization}\label{ssec:opt_analogy}

There is an interesting analogy between our work on Fisher information guarantees for sampling and the study of finding approximate stationary points in optimization.
As noted in Section~\ref{ssec:motivation}, it is well-established that gradient descent finds an $\varepsilon$-stationary point in $O(1/\varepsilon^2)$ iterations, and the lower bound of~\cite{Car+20StatPt} shows that this is the best possible, at least in high dimension.
Recent investigations, however, have focused on the complexity in low dimension, where the picture is more nuanced.

To motivate this study, consider first the problem of convex, non-smooth minimization, which is a separate problem altogether but shares some similar features with the stationary point problem.
In particular, gradient descent finds an $\varepsilon$-minimizer in $O(1/\varepsilon^2)$ steps, and lower bounds show that this is optimal in high dimension.
In low dimension, however, there is a better algorithm: cutting plane methods can find an $\varepsilon$-minimizer in $O(d\log(1/\varepsilon))$ steps.
At the level of lower bounds, this manifests in the following way: the lower bound construction in~\cite{Car+20StatPt} showing that $O(1/\varepsilon^2)$ is tight is embedded in dimension $d \gg 1/\varepsilon^2$.
The existence of the cutting plane method means that no such constructions exist when $d \log(1/\varepsilon) \ll 1/\varepsilon^2$, i.e., in dimension $d \ll 1/(\varepsilon^2 \log(1/\varepsilon))$.

Returning to the stationary point question, the interesting departure from the paragraph above is that thus far, we have not discovered an analogue of cutting plane methods for non-convex optimization.
Indeed, the best complexities in low dimension scale as $O(1/\varepsilon^{2-O(1/d)})$~\cite{BubMik20GradFlow}.
If this is fundamental, then there should be lower bound constructions witnessing the $\Omega(1/\varepsilon^2)$ rate which can be embedded in much smaller dimension: say, $d \asymp \log(1/\varepsilon)$ or so.
However, the lower bound construction of~\cite{Car+20StatPt} requires a much larger dimension: $d \gg 1/\varepsilon^2$ for deterministic methods, and $d \gg \widetilde{\Omega}(1/\varepsilon^4)$ for randomized methods.
Thus, at present, it is unclear whether or not there could exist a cutting plane method for finding stationary points.

More recently,~\cite{CheBubSal23Stat} showed that in dimension one, when we restrict to deterministic algorithms which only use evaluations of $\nabla f$ (and not to evaluations of $f$), then the $O(1/\varepsilon^2)$ complexity is already tight; moreover, in dimension two,~\cite{HolZam23Stat} showed that for deterministic algorithms with access to both $f$ and $\nabla f$, the optimal complexity is $\Theta(1/\varepsilon)$.
For general algorithms (randomized, with access to both $f$ and $\nabla f$), the optimal complexity is $\Theta(\log(1/\varepsilon))$ in dimension one~\cite{CheBubSal23Stat}, $\Theta(1/\varepsilon^2)$ in large dimension $\gg \widetilde\Omega(1/\varepsilon^4)$ (as discussed above), and remains open in between.

As for the connection with our work: although the low-dimensional complexity of finding stationary points remains elusive in optimization, \emph{our work resolves the analogous question in the sampling setting}.
Indeed, the lower bound $\Omega(1/\varepsilon^{2-o(1)})$ in~\cite{Che+23FisherLower} already holds in dimension $\gg \sqrt{\log(1/\varepsilon)/\log\log(1/\varepsilon)}$, and the present work shows that $\widetilde O(1/\varepsilon^2)$ can be achieved via the proximal sampler in any dimension $d \ge 1$.
Although this does not have direct implications for the optimization question, which remains open, we believe that this further strengthens the rich connections between the two fields.

\subsection{Additional related work}

First-order stationarity guarantees have also been investigated for the Stein variational gradient descent (SVGD) algorithm, with initial works focusing on the kernelized Stein discrepancy (KSD), a kernelized analogue of the relative Fisher information~\cite{Liu17SVGD, SalSunRic22SVGD, ShiMac23SVGD, SunKarRic23SVGD, BalBanGho25SVGD}.
More recently, relative Fisher information guarantees were obtained for a regularized variant of SVGD~\cite{He+25RSVGF, He+26RSVGD}.

As discussed above,~\cite{Che+23FisherLower} established a lower bound on finding approximate stationary points for sampling which scales as $\Omega((L/\varepsilon^2)^{1-o(1)})$ in high dimension. In addition, that paper established a second lower bound which shows that the complexity of obtaining a moderate bound of the form $\FI(\cdot\mmid\pi) \lesssim Ld$ can be precisely related to the complexity of finding approximate stationary points in non-convex optimization. In particular, it implies that the Fisher information guarantee for averaged LMC in~\cite{balasubramanian22a} is already optimal in a certain regime. However, this regime does not cover guarantees $\FI(\cdot\mmid\pi)\le \varepsilon^2$ for small $\varepsilon \ll 1$, which is the regime of interest in this work. Moreover, their reduction is distinct from the ones we present in Theorem~\ref{thm:main} and~\ref{thm:converse}, which relate the complexity of obtaining small Fisher information to the complexity of high-accuracy log-concave sampling.

\section{A review of the proximal sampler}
\label{ssec:prox_sampler}

Throughout, we are working with standard notations and definitions, which we review in Appendix~\ref{Sec:Def}.

We review the proximal sampler algorithm, originally proposed in~\cite{LST21structured}.
Our goal is to draw a sample from the target distribution $\pi^X \propto \exp(-f)$ supported on $\R^d$.
Let $h > 0$ be an arbitrary ``step size''.
We define the joint probability distribution $\pi$ on $\R^{2d}$ with density function:
$$\pi^{XY}(x,y) \propto \exp\left( -f(x) - \frac{1}{2h} \|x-y\|^2 \right).$$
Observe that the $x$-marginal of $\pi$ is precisely the original target distribution: $\pi^X(x) = \int_{\R^d} \pi^{XY}(x,\d y)$.
Therefore, to sample from $\pi^X$, it suffices to sample $(X,Y) \sim \pi^{XY}$ and only return the $X$ component.

To sample from $\pi^{XY}$,~\cite{LST21structured} proposed to run Gibbs sampling with an alternating update.
This results in the \textit{ideal} proximal sampler algorithm, where in each iteration $k \ge 0$, from the current iterate $x_k \in \R^d$, we perform the following two steps:
\begin{enumerate}
    \item \textbf{Forward step:} Draw $y_k \mid x_k \sim \pi^{Y \mid X}( \cdot \mid x_k) = \N(x_k, hI)$.
    \item \textbf{Backward step:} Draw $x_{k+1} \mid y_k \sim \pi^{X \mid Y}(\cdot \mid y_k)$.
\end{enumerate}
By the general theory of Gibbs sampling, the Markov chain update from $(x_k, y_k)$ to $(x_{k+1}, y_{k+1})$ is reversible with respect to the joint distribution $\pi^{XY}$. 

We call the above the ideal proximal sampler since we assume we can implement both the forward and backward steps exactly.
The forward step is easy to implement, since it only requires drawing a Gaussian distribution around the current iterate.
The backward step is called the \textit{restricted Gaussian oracle (RGO)} in~\cite{LST21structured}, and it is non-trivial since it requires sampling from the following probability distribution over $x \in \R^d$, for any fixed $y \in \R^d$:
\begin{align}\label{Eq:BackwardCond}
    \pi^{X \mid Y}(x \mid y) \propto_x \exp\left( -f(x) - \frac{1}{2h} \|x-y\|^2 \right).
\end{align}
Note the proportionality is only over $x \in \R^d$ (with fixed $y \in \R^d$), which we denote with $\propto_x$.

\subsection{Description of proximal sampler at the level of distributions}

For $h > 0$, let $\Q^h = (\Q^h_y)_{y \in \R^d}$ denote the Gaussian channel with variance $h$, which acts on a probability distribution $\rho$ by:
$$\rho \Q^h  \deq  \rho \ast \N(0, hI).$$
Thus, $\Q^h_y = \N(y, hI)$ for each $y \in \R^d$.

Let $\RGO^h = (\RGO^h_y)_{y \in \R^d}$ denote the RGO channel with step size $h > 0$ formed by $\pi^X \propto \exp(-f)$, so:
$$\RGO^h_y = \pi^{X \mid Y}(\cdot \mid y)$$
where $\pi^{X \mid Y}(\cdot \mid y)$ is the conditional distribution given in~\eqref{Eq:BackwardCond}.
In our subsequent discussion below, when the step size $h > 0$ is fixed, we write $\RGO \equiv \RGO^h$ for brevity.

We can describe the ideal proximal sampler as follows.
In each iteration $k \ge 0$, from the current iterate $x_k \sim \rho_k^X$, we perform the forward step to draw $y_k \sim \rho_k^Y$ given by:
$$\rho_k^Y = \rho_k^X \Q^h \,,$$
and then we perform the backward step to draw $x_{k+1} \sim \rho_{k+1}^X$ given by:
$$\rho_{k+1}^X = \rho_k^Y \RGO^h \,.$$
Thus, one step of the ideal proximal sampler is a composition of the Gaussian channel with the RGO channel:
$\rho_{k+1}^X = \rho_k^X \Q^h \RGO^h \,.$

From the target distribution $\pi^X$, we define
$$\pi^Y  \deq  \pi^X \Q^h = \pi^X \ast \N(0, hI) \,.$$
Then, since $\pi^X$ is invariant under one step of the proximal sampler, we have: 
$\pi^X = \pi^X \Q^h \RGO^h = \pi^Y \RGO^h \,.$

In Section~\ref{Sec:IdealProxSampler} we study the guarantee of the proximal sampler with an exact implementation of the RGO.
In Section~\ref{Sec:ApproxProxSampler} we study the guarantee of the proximal sampler when we use an approximate RGO implementation $\hRGO^h$ in place of the exact RGO $\RGO^h$, with the benefit of obtaining an improved dimension dependence.

\section{Guarantees for the ideal proximal sampler}
\label{Sec:IdealProxSampler}

In this section, we provide guarantees for the \textit{ideal} proximal sampler algorithm, which means we assume an exact implementation of the RGO.

\subsection{Guarantees for one step of the proximal sampler}

We have the following guarantee on one step of the proximal sampler.
We provide the proof of Lemma~\ref{Lem:ProxExact} in Appendix~\ref{Sec:LemProxExactProof}.

\begin{lemma}\label{Lem:ProxExact}
    Assume $\pi^X \propto \exp(-f)$ is $L$-log-smooth for some $L \in (0, \infty)$.
    Let $h \in (0, \frac{1}{L})$.
    Then we have the following.
    \begin{enumerate}
        \item \textbf{Forward step:} 
        For any probability distribution $\rho^X$, let $\rho^Y = \rho^X \Q^h$. 
        Then:
        \begin{align}\label{Eq:FwdStep}
            \frac{h\,(1-hL)}{2} \, \FI(\rho^Y \,\|\, \pi^Y) &\le \KL(\rho^X \,\|\, \pi^X) - \KL(\rho^Y \,\|\, \pi^Y).
        \end{align}

        \item \textbf{Backward step:} 
        For any probability distribution $\rho^Y$, let $\rho^X_+ = \rho^Y \RGO^h$. 
        Then:
        \begin{align}
            \FI(\rho^X_+ \,\|\, \pi^X)
            &\le \min \Bigl\{ \FI(\rho^Y \,\|\, \pi^Y) \,,\;  \frac{2}{h} \left(\KL(\rho^Y \,\|\, \pi^Y) - \KL(\rho^X_+ \,\|\, \pi^X)\right) \Bigr\} \,.\label{Eq:BwdStep}
        \end{align}
    \end{enumerate}
\end{lemma}

\subsection{Convergence guarantee for the ideal proximal sampler}

By iterating Lemma~\ref{Lem:ProxExact}, we show the following convergence guarantee for the ideal proximal sampler, assuming an exact implementation of the RGO.
We provide the proof of Theorem~\ref{Thm:FIExactRGO} in Appendix~\ref{Sec:FIExactRGOProof}.

\begin{theorem}\label{Thm:FIExactRGO}
    Assume that $\pi^X \propto \exp(-f)$ on $\R^d$ is $L$-log-smooth for some $L \in (0,\infty)$.
    Let $h \in (0, \frac{1}{L})$.
    Suppose we run the ideal proximal sampler algorithm (with an exact RGO implementation) from any $X_0 \sim \rho_0^X$ to get iterates $X_k \sim \rho_k^X$ for $k \ge 1$.
    Then for all $K \in \bN$, we have:
    \begin{align}\label{Eq:BoundTotal}
        \sum_{k=1}^K \FI(\rho_k^X \,\|\, \pi^X) \le \frac{\KL(\rho_0^X \,\|\, \pi^X)}{h\,\bigl(1-\frac{Lh}{2}\bigr)}\,.
    \end{align}
    In particular, the average iterate $\bar \rho_K^X \deq  \frac{1}{K} \sum_{k=1}^K \rho_k^X$ satisfies:
    \begin{align}\label{Eq:BoundAvg}
        \FI(\bar \rho_K^X \,\|\, \pi^X) \le \frac{\KL(\rho_0^X \,\|\, \pi^X)}{h\,(1-\frac{Lh}{2})\, K}\,.
    \end{align}    
\end{theorem}

\subsection{Rejection sampling implementation of the RGO}

For $h \asymp 1/(Ld)$,~\cite[Section~4.2]{CCSW22improved} showed that the RGO can be implemented exactly via rejection sampling using one evaluation of $\prox_{hf}$ and $O(1)$ density evaluations in expectation. This immediately leads to the following result.

\begin{corollary}\label{Cor:FIExactRGOComplexity}
    Assume that $\pi^X \propto \exp(-f)$ on $\R^d$ is $L$-log-smooth for some $L \in (0,\infty)$.
    Choose $h \asymp \frac{1}{Ld}$.
    Then, the proximal sampler with an exact rejection sampling implementation of the RGO yields a sample from a distribution $\bar\rho$ satisfying $\FI(\bar\rho\mmid \pi^X) \le \varepsilon^2$ using
    \begin{align*}
        O\Bigl(\frac{Ld\,\KL(\rho_0^X \mmid \pi^X)}{\varepsilon^2}\Bigr)\qquad\text{queries to}~\prox_{hf}~\text{and}~f~\text{in expectation}\,.
    \end{align*}
\end{corollary}

Denoting $\KL_0 \deq \KL(\rho_0^X \mmid \pi^X)$, this already substantially improves upon the $O(L^2 d\,\KL_0/\varepsilon^4)$ iteration complexity obtained in~\cite{balasubramanian22a}. However, there is further room for improvement: several subsequent works have developed improved RGO implementations with better iteration complexities~\cite{FanYuaChe23ImprovedProx, AltChe24Warm, Chen+26HighAccDiffusion}. Following these trends, in the next section, we develop improved bounds based on approximate RGO implementation.

\section{Guarantees for the proximal sampler with an approximate RGO implementation}
\label{Sec:ApproxProxSampler}

\subsection{Main result}

The bound in Theorem~\ref{Thm:FIExactRGO} assumes that in each iteration, we can implement the RGO in the backward process exactly.
Although this can be accomplished via rejection sampling, it does not yield the best dimension dependence.
Here, we show that implementation of the RGO can be reduced to the complexity of high-accuracy \emph{log-concave} sampling.

\begin{assumption}\label{ass:sampler}
    There exists a sampling algorithm  $\Alg$ with the following property: for any $1$-strongly log-concave and $3$-log-smooth target distribution $\pi$ with mode at the origin, and tolerance parameter $\delta > 0$, $\Alg$ returns a sample with distribution $\widehat\pi$ such that $\Ren_3(\widehat\pi \mmid \pi) \le \delta^2$.
\end{assumption}

The key point is that with step size $h=1/(2L)$, the RGO distribution $\RGO_y$ is $L$-strongly log-concave, $3L$-log-smooth, with mode at $\prox_{f/(2L)}(y)$.
By translating and rescaling, it follows that $\RGO_y$ can be implemented up to $\Ren_3$ error $\delta^2$ using one call to $\prox_{f/(2L)}(y)$ and one call to $\Alg$ with tolerance parameter $\delta$.

We remark that the assumption that we can evaluate $\prox_{f/(2L)}$ is not restrictive: with our choices of parameters, the computation of $\prox_{f/(2L)}$ is a strongly convex and smooth optimization problem and can be solved to high accuracy via standard convex optimization algorithms.

Our main reduction is stated below.
By the ``proximal sampler with smoothed RGO implementation'', we mean that the proximal sampler is run where the RGO is implemented by taking the output of $\Alg$ on the RGO and adding a Gaussian of appropriate variance to the output; finally, output one of the proximal sampler iterates chosen uniformly at random. See Algorithm~\ref{alg:smoothed-proximal-sampler} for a precise description. 

\begin{algorithm}[ht]
\caption{Proximal Sampler with Smoothed RGO Implementation}\label{alg:smoothed-proximal-sampler}
\begin{algorithmic}[1]
\Require Target $\pi^X \propto \exp(-f)$ on $\mathbb{R}^d$, smoothness $L$, initial point $X_0 \sim \rho_0^X$, number of steps $K$, RGO accuracy $\delta$, smoothing variance $t>0$.
\For{$k = 0,1,\dots,K-1$}
    \State Sample $Y_k \sim \N(X_k, \frac{1}{2L}\,I_d)$.
    \State Compute the RGO mode $m_k \gets \operatorname{prox}_{f/(2L)}(Y_k)$.
    \State Define the rescaled RGO target on $\mathbb{R}^d$
    \[
        \nu_k(u)
        \propto
        \exp\Bigl(
            -f\bigl(m_k + \frac{u}{\sqrt L}\bigr)
            - L\,\Bigl\|m_k + \frac{u}{\sqrt L} - Y_k\Bigr\|^2
        \Bigr)\,.
    \]
    \State Run $\Alg$ with accuracy $\delta$ on $\nu_k$ to obtain $U_k \sim \widehat{\nu}_k$.
    \State Set the approximate RGO sample $\widetilde X_{k+1} \gets m_k + U_k/\sqrt L$.
    \State Sample $X_{k+1} \sim \N(\widetilde X_{k+1}, tI_d)$.
\EndFor
\State Sample $J \sim \operatorname{Unif}\{1,\dots,K\}$
\State \Return $X_J$
\end{algorithmic}
\end{algorithm}

\begin{theorem}\label{thm:main}
    Assume that the target distribution $\pi^X \propto \exp(-f)$ on $\R^d$ is $L$-smooth.
    Let $\rho_0^X$ be an initial distribution with $\KL_0 \deq \KL(\rho_0^X \mmid \pi^X)$ and $\Delta_0 \deq \Ren_2(\rho_0^X \mmid \pi^X)$.
    Then, under Assumption~\ref{ass:sampler}, Algorithm~\ref{alg:smoothed-proximal-sampler} returns a sample from $\widehat\pi^X$ with $\FI(\widehat\pi^X \mmid \pi^X) \le \varepsilon^2$ using $O(L\,\KL_0/\varepsilon^2)$ queries to $\prox_{f/(2L)}$ and calls to $\Alg$ with accuracy parameter $\delta = \poly(1/\Delta_0, \varepsilon^2/L, 1/d)$ and smoothing variance $t \asymp (\delta^2 + d\sqrt{\delta})^{1/4}/L$.
\end{theorem}

We provide a proof sketch of Theorem~\ref{thm:main} in Section~\ref{ssec:pf_outline}, and the full proof in Appendix~\ref{Sec:ProofThmMain}.

Current state-of-the-art samplers (e.g.,~\cite{FanYuaChe23ImprovedProx, AltChe24Warm, Chen+26HighAccDiffusion}) satisfy Assumption~\ref{ass:sampler}, where each run of $\Alg$ uses $\widetilde O(\sqrt d\polylog(1/\delta))$ queries to $\nabla f$ in expectation.
We therefore conclude:

\begin{corollary}
    Assume that the target distribution $\pi^X \propto \exp(-f)$ on $\R^d$ is $L$-smooth.
    Let $\rho_0^X$ be an initial distribution with $\KL_0 \deq \KL(\rho_0^X \mmid \pi^X)$ and $\Delta_0 \deq \Ren_2(\rho_0^X \mmid \pi^X)$.
    Then, Algorithm~\ref{alg:smoothed-proximal-sampler} with $\Alg$ given by~\cite{Chen+26HighAccDiffusion} returns a sample from $\widehat\pi^X$ with $\FI(\widehat \pi^X \mmid \pi^X) \le \varepsilon^2$ using
    \begin{align*}
        O\biggl(\frac{L\sqrt d\,\KL_0}{\varepsilon^2} \polylog(\Delta_0, L/\varepsilon^2, d)\biggr) \qquad\text{queries to}~\prox_{f/(2L)}~\text{and}~\nabla f~\text{in expectation}\,.
    \end{align*}
\end{corollary}

Recall that~\cite{Che+23FisherLower} gave a lower bound construction with $\KL_0 \le 1$ for which the complexity of obtaining $\sqrt{\FI} \le \varepsilon$ is $\widetilde\Omega((L/\varepsilon^2)^{1-o(1)})$ for all $d \ge \widetilde\Omega(\sqrt{\log(1/\varepsilon)})$. Interestingly, our reduction in Theorem~\ref{thm:main} implies that if this lower bound can be strengthened to exhibit polynomial dependence on $d$, then it would imply a lower bound on the complexity of high-accuracy \emph{log-concave} sampling which also exhibits polynomial dependence on $d$. This would be quite significant as the current best lower bound in that setting only exhibits a dimension dependence of $\Omega(\log d)$~\cite{Chewi+24QueryLower}. Thus, our reduction suggests an alternative approach for this problem which could be easier, e.g., via enabling the use of non-convex bump function constructions.

\subsection{Converse reduction}

Here, we argue that the reduction to high-accuracy log-concave sampling in Theorem~\ref{thm:main} is the right one, by presenting a reduction in the converse direction via a simple restart strategy.
We provide the proof of Theorem~\ref{thm:converse} in Appendix~\ref{Sec:ThmConverseProof}.

\begin{theorem}\label{thm:converse}
    Assume that there is an algorithm $\Alg$ with the following property: for any $L$-log-smooth target distribution $\pi$ and an initialization $\rho_0$ satisfying $\KL(\rho_0 \mmid \pi) \le \KL_0$, $\Ren_2(\rho_0\mmid \pi) \le \Delta_0$, $\Alg$ outputs a sample from a distribution $\widehat\pi$ with $\Ren_2(\widehat\pi\mmid\pi) \le \Delta_0$ and $\FI(\widehat\pi\mmid\pi) \le \varepsilon^2$ using $\frac{L\,\KL_0}{\varepsilon^2}\,\phi(d) \polylog(\Delta_0, L/\varepsilon^2)$ queries to an oracle for $\pi$.

    Then, there is another algorithm $\Alg'$ with the following property: for any $\alpha$-strongly log-concave and $L$-log-smooth distribution $\pi$ with mode at the origin, $\Alg'$ outputs a sample from a distribution $\widehat\pi$ with $\KL(\widehat\pi\mmid\pi) \le \varepsilon^2$ using $O(\kappa\,\phi(d) \polylog(d, \kappa, \frac{1}{\varepsilon}))$ queries to the same oracle for $\pi$, where $\kappa \deq \frac{L}{\alpha}$ is the condition number.
\end{theorem}

Note that, aside from the technical condition that $\Alg$ satisfies $\Ren_2(\widehat\pi\mmid \pi) \le \Delta_0$, the conclusion of Theorem~\ref{thm:main} satisfies the requirement of $\Alg$ with $\phi(d) = \widetilde O(\sqrt d)$, which by Theorem~\ref{thm:converse} leads to a high-accuracy log-concave sampler with complexity $\widetilde O(\kappa \sqrt d)$. Morally, the conclusion is that the dimension dependence in Theorem~\ref{thm:main} cannot be improved without also improving the dimension dependence of high-accuracy log-concave sampling.

However, we note that there is a slight gap: Theorem~\ref{thm:main} assumes the existence of a $\Ren_3$ sampler, whereas the output of Theorem~\ref{thm:converse} is only a $\KL$ sampler. We do not view this issue as fundamental, as there seems to not be any reason to expect the complexity of sampling in $\KL$ vs.\ $\Ren_q$ to be substantially different: indeed, guarantees which are initially proven in $\KL$ are often later improved to hold in $\Ren_q$. Moreover, a $\KL$ guarantee can usually be boosted to a $\Ren_3$ guarantee, provided that we have a warm start in a slightly stronger R\'enyi divergence such as $\Ren_4$; see the argument of~\cite[Lemma 5.2]{AltChe24Warm}. Therefore, Theorem~\ref{thm:converse} should also imply the existence of a $\Ren_3$ sampler with complexity $\widetilde O(\kappa\phi(d))$ from a suitable warm start.

\subsection{Proof outline of Theorem~\ref{thm:main}}\label{ssec:pf_outline}

Recall that for each $y \in \R^d$ and step size $h > 0$, the backward conditional distribution is
$$\pi^{X \mid Y}(x \mid y) \propto \exp\Bigl(-f(x) - \frac{1}{2h}\, \|x-y\|^2\Bigr)\,.$$
Recall the RGO channel $\RGO$ is the map via this backward conditional distribution:
For any $\rho^Y$, the distribution $\rho^Y \RGO$ is given by
$$(\rho^Y \RGO)(x) = \int_{\R^d} \pi^{X \mid Y}(x \mid y) \, \rho^Y(\dy)\,.$$

\subsubsection{Key lemma: Guarantee for the proximal sampler with approximate RGO}

The following key lemma proves a bound on the relative Fisher information along one step of the proximal sampler with an inexact implementation of the RGO.
We provide the proof of Lemma~\ref{Lem:FIApproxRGO} in Appendix~\ref{Sec:LemFIApproxRGOProof}.

\begin{lemma}\label{Lem:FIApproxRGO}
    Assume $\pi^X \propto \exp(-f)$ on $\R^d$ is $L$-log-smooth for some $L \in (0, \infty)$.
    Let $\RGO \equiv \RGO^h$ be the RGO channel with step size $h \in (0, \frac{1}{L})$, and let $\hRGO \equiv \hRGO^h$ be an approximate RGO implementation.
    Consider one step of the proximal sampler algorithm with the approximate RGO.
    Namely, from any input distribution $\rho^X$, define $\rho^Y \deq \rho^X \Q^h = \rho^X \ast \N(0, hI)$, and define $\rho^X_+ \deq \rho^Y \hRGO$.
    If $\hRGO$ satisfies the following error guarantees for some $\bar\err_\KL, \bar\err_\FI \ge 0$,
    \begin{align}
        \KL(\rho^Y \hRGO \,\|\, \pi^X) &\le \KL(\rho^Y \RGO \,\|\, \pi^X) + \bar\err_\KL^2\,, \label{Eq:ApproxGuaranteeKL} \\
        \FI(\rho^Y \hRGO \,\|\, \pi^X) &\le
        2\,\FI(\rho^Y \RGO \,\|\, \pi^X) + \bar\err_\FI^2\,, \label{Eq:ApproxGuaranteeFI}
    \end{align}
    then we have
    \begin{align*}
        \frac{h}{2}\,\left(1-\frac{Lh}{2}\right)\,\FI(\rho_+^X\,\|\,\pi^X)
        &\le \KL(\rho^X\,\|\,\pi^X) - \KL(\rho_+^X\,\|\,\pi^X) + \bar\err_\KL^2 + \frac{h}{2}\,\left(1-\frac{Lh}{2}\right)\,\bar\err_\FI^2\,.
    \end{align*}
\end{lemma}

\subsubsection{Applying the key lemma}

In order to use Lemma~\ref{Lem:FIApproxRGO}, we must control the quantities $\bar\err_\KL$ and $\bar\err_\FI$.
Since the KL divergence does not satisfy a triangle inequality, we cannot simply bound $\bar\err_\KL^2$ by $\KL(\rho^Y\hRGO\mmid \rho^Y\RGO)$; however, we can establish the following substitute.
We provide the proof of Lemma~\ref{lem:KL_bd} in Appendix~\ref{Sec:LemKL_bdProof}.

\begin{lemma}\label{lem:KL_bd}
    For any probability distributions $\mu$, $\nu$, and $\pi$ with $\mu \ll \nu \ll \pi$:
    \begin{align*}
        \KL(\mu\mmid \pi) - \KL(\nu \mmid \pi) \le \KL(\mu\mmid \nu) + (2+\Ren_2(\nu\mmid\pi))\sqrt{\chi^2(\mu\mmid\nu)}\,.
    \end{align*}
\end{lemma}

This lemma implies we can choose the quantity $\bar\err_\KL^2$ in~\eqref{Eq:ApproxGuaranteeKL} to be
\begin{align}
    \bar\err_\KL^2
    &\stackrel{(i)}{\le} \KL(\rho^Y \hRGO \mmid \rho^Y\RGO) + (2+\Ren_2(\rho^Y\RGO\mmid \pi^X))\sqrt{\chi^2(\rho^Y \hRGO \mmid \rho^Y\RGO)} \notag \\
    &\stackrel{(ii)}{\le} \sup_{y\in\R^d} \KL(\hRGO_y \mmid \RGO_y) + (2+\Ren_2(\rho^Y\mmid \pi^Y)) \sup_{y\in\R^d}\sqrt{\chi^2(\hRGO_y \mmid \RGO_y)} \notag \\
    &\stackrel{(iii)}{\le} \err_{\chi^2}^2 + (2+\Ren_2(\rho^Y\mmid \pi^Y)) \,\err_{\chi^2}\,,\label{eq:KL_bd}
\end{align}
with $\err_{\chi^2}^2 \deq \sup_{y\in\R^d}\chi^2(\hRGO_y \mmid \RGO_y)$.
In the above, the first inequality~$(i)$ is by applying Lemma~\ref{lem:KL_bd}; the second inequality~$(ii)$ follows by the joint convexity of KL divergence, which implies $\KL(\rho^Y \hRGO \mmid \rho^Y\RGO) \le \E_{y \sim \rho^Y}[\KL(\hRGO_y \mmid \RGO_y)] \le \sup_{y\in\R^d} \KL(\hRGO_y \mmid \RGO_y)$, and similarly for chi-square divergence; and the third inequality~$(iii)$ follows from the bound $\KL(\hRGO_y \mmid \RGO_y) \le \chi^2(\hRGO_y \mmid \RGO_y)$.

For the quantity $\bar\err_{\FI}$ in~\eqref{Eq:ApproxGuaranteeFI}, again the triangle inequality fails to hold, and moreover there is an additional difficulty because it is not trivial to bound $\FI(\rho^Y\hRGO \mmid \rho^Y \RGO)$ in terms of $\sup_{y\in\R^d} \FI(\hRGO_y \mmid \RGO_y)$.
Here, we are able to show the following lemma.
We provide the proof of Lemma~\ref{lem:FI_bd} in Appendix~\ref{Sec:LemFI_bdProof}.

\begin{lemma}\label{lem:FI_bd}
    Assume $\pi^X \propto \exp(-f)$ is $L$-log-smooth for some $L \in (0,\infty)$.
    Let $\RGO \equiv \RGO^h$ be the RGO channel with step size $h > 0$, and let $\hRGO \equiv \hRGO^h$ be an approximate RGO channel with accuracies $\err_{\FI}$ and $\err_{\KL} \ll 1$ in relative Fisher information and KL divergence, respectively: 
    \begin{align*}
        \sup_{y\in\R^d} \FI(\hRGO_y \mmid \RGO_y) &\le \err_{\FI}^2 \,, \\
        \sup_{y\in\R^d} \KL(\hRGO_y \mmid \RGO_y) &\le \err_{\KL}^2 \ll 1 \,.
    \end{align*}
    Then for any probability distribution $\rho^Y$ on $\R^d$, it holds that:
\begin{align*}
    \FI(\rho^Y \hRGO \mmid \pi^X) - 2\,\FI(\rho^Y \RGO\mmid \pi^X)
    \lesssim \left(L+ \frac{1}{h} \right) \left(d+\log \frac{1}{\err_{\KL}} + \Ren_2(\rho^Y \mmid \pi^Y) \right) \,\err_\KL + \err_\FI^2\,.
\end{align*}
\end{lemma}

To summarize, in order to apply Lemma~\ref{Lem:FIApproxRGO}, we must control (1) the R\'enyi-$2$ divergence (equivalently, the chi-squared divergence) along the iterates of the proximal sampler, and (2) the accuracies $\err_\KL$, $\err_{\chi^2}$, $\err_{\FI}$ in implementing the RGO\@.
We handle (1) in Lemma~\ref{lem:chisq_prox_sampler} in Appendix~\ref{Sec:Lemchisq_prox_sampler}.

For (2), the main difficulty is that we wish to implement the RGO using high-accuracy samplers (e.g., MALA), but standard guarantees for such samplers are stated in terms of $\KL$ or $\chi^2$, rather than $\FI$.
Rather than establishing separate $\FI$ guarantees for each high-accuracy sampler, we show via a generic smoothing argument that it suffices to implement the RGO in $\Ren_3$.
We provide the proof of Corollary~\ref{cor:FI_from_chisq} in Appendix~\ref{Sec:CorFI_from_chisq_proof}.

\begin{corollary}\label{cor:FI_from_chisq}
    Assume $\pi^X \propto \exp(-f)$ is $L$-log-smooth for some $L \in (0,\infty)$.
    Let $\RGO \equiv \RGO^h$ be the RGO channel with step size $h = \frac{1}{2L}$.    
    Suppose that there is an approximate RGO implementation $\widetilde{\msf R}$ with the guarantee that
    \begin{align*}
        \sup_{y\in\R^d} \chi^2(\widetilde{\msf R}_y \mmid \RGO_y) \le \err_{\chi^2}^2 \ll 1/d^4\,.
    \end{align*}
    Then there is an approximate implementation $\hRGO$ which satisfies  for all probability distributions $\rho^Y$:
    \begin{align*}
        \FI(\rho^Y \hRGO \mmid \pi^X) - 2\,\FI(\rho^Y \RGO\mmid \pi^X)
        \lesssim L \left(d+\log(1/\err_{\chi^2}) + \Ren_2(\rho^Y\mmid \pi^Y) \right)\,\left\{d(d+\log(1/\err_{\chi^2}))\,\err_{\chi^2} \right\}^{1/4}\,.
    \end{align*}
    Furthermore, if the approximate implementation $\widetilde{\msf R}$ satisfies
    \begin{align*}
        \sup_{y\in\R^d}\Ren_3(\widetilde{\msf R}_y \mmid \RGO_y) \le \err_{\Ren_3}^2 \ll 1\,,
    \end{align*}
    then we can also ensure that $\hRGO$ satisfies
    \begin{align*}
        \sup_{y\in\R^d} \chi^2(\hRGO_y\mmid \RGO_y) \lesssim \err_{\Ren_3}^2 + d\sqrt{\err_{\chi^2}\,\left(1 + \frac{\log(1/\err_{\chi^2})}{d}\right)}\,.
    \end{align*}
\end{corollary}

In fact, the new approximate RGO $\hRGO$ is easy to implement; given an approximate implementation $\widetilde{\msf R}_y$ (which we obtain from Assumption~\ref{ass:sampler}), the new approximate RGO $\hRGO$ is given by $\hRGO_y = \widetilde{\msf R}_y \Q^t = \widetilde{\msf R}_y \ast \mathcal N(0, tI)$ with an appropriately chosen smoothing parameter $t > 0$;
practically, given a sample $X \sim \widetilde{\msf R}_y$, a sample from $\hRGO_y$ is given by $X + \sqrt{t}Z$ for an independent $Z \sim \N(0,I)$.

Combining together all of these ingredients and performing some bookkeeping yields Theorem~\ref{thm:main}.
We provide the full proof in Appendix~\ref{Sec:ProofThmMain}.

\section{Conclusion}

In this work, we obtained a new complexity bound for sampling from a measure with small relative Fisher information to the target, which can be interpreted as an approximate first-order stationary measure for sampling. Our bound, which is based on the proximal sampler, is near-optimal in its dependence on the precision $\varepsilon$. Moreover, the dimension dependence of our bound matches the dimension dependence of high-accuracy log-concave sampling, and the former cannot be improved without also improving the latter. This provides a relatively complete picture of the Fisher information complexity of sampling.

Going forward, one interesting open problem is to improve the Fisher information lower bound of~\cite{Che+23FisherLower} to exhibit polynomial dependence on the dimension since, by our results, it would imply a corresponding lower bound with polynomial dimension dependence for high-accuracy log-concave sampling.
Another open problem is to determine the optimal Fisher information complexity in dimension one (see~\cite{Che+23FisherLower}).

\newpage
\appendix

\section{Notation and definitions}
\label{Sec:Def}


We are working in the Euclidean space $\R^d$ of dimension $d \ge 1$.
We denote the $\ell_2$-inner product by $\langle u,v \rangle$ and the $\ell_2$-norm by $\|u\| \deq \sqrt{\langle u,u \rangle}$ for $u,v \in \R^d$.
For $x \in \R$, let $\{x\}_+ = \max\{x,0\}$ and $\{x\}_- = -\min\{x,0\}$.
We use the notation $a \lesssim b$ (resp.\ $a \gtrsim b$) to mean that $a\le Cb$ (resp.\ $a \ge Cb$) for a universal constant $C > 0$, and $a \asymp b$ to mean $a\lesssim b$ and $a \gtrsim b$.
If a statement holds provided that $a\ge Cb$ for a sufficiently large universal constant, then we abbreviate this by writing ``if $a \gg b$, then \ldots''.

For a symmetric matrix $A \in \R^{d \times d}$, we denote its operator norm by $\|A\|_\op$, which is the largest absolute value of the eigenvalues of $A$.
We denote the Hilbert-Schmidt (or Frobenius) norm of $A \in \R^{d \times d}$ by $\|A\|_{\HS} = \sqrt{\Tr(A^\top A)}$.
For a symmetric matrix $A \in \R^{d \times d}$, we say $A \succeq 0$ is positive semidefinite if $u^\top Au \ge 0$ for any $u \in \R^d$.
For symmetric matrices $A,B \in \R^{d \times d}$, $A \succeq B$ if $A-B \succeq 0$.
Throughout, $I_d \in \R^{d \times d}$ denotes the identity matrix.

For a function $f \colon \R^d \to \R$, we say $f$ is $L$-smooth if $f$ is twice-differentiable and $\|\nabla^2 f(x)\|_\op \le L$ for all $x \in \R^d$.
We recall the \textit{prox operator} of $f$ is the function $\prox_f \colon \R^d \to \R^d$ given by, for any $y \in \R^d$:
$$\prox_f(y) = \arg\min_{x \in \R^d} \left\{ f(x) + \frac{1}{2}\, \|x-y\|^2 \right\}$$
assuming the inner objective function on the right-hand side above has a minimizer ($\prox_f$ can return any minimizer, if not unique).
Note that for any $h > 0$, we can write $\prox_{hf}$ as:
$$\prox_{hf}(y) = \arg\min_{x \in \R^d} \left\{ f(x) + \frac{1}{2h}\, \|x-y\|^2 \right\}.$$
If $f$ is $L$-smooth and $h < \frac{1}{L}$, then the inner objective above is a strictly convex function and has a unique minimizer, so $\prox_{hf}$ is well-defined.

For two measures $\mu, \nu$ on $\R^d$, we say $\mu$ is absolutely continuous with respect to $\nu$, denoted $\mu \ll \nu$, if $\nu(A) = 0$ implies $\mu(A) = 0$ for any measurable subset $A \subset \R^d$.
If a probability measure $\mu$ is absolutely continuous with respect to the Lebesgue measure $\dx$ on $\R^d$, then we can represent $\mu$ in terms of its probability density function, which we also denote by $\mu \colon \R^d \to [0,\infty)$ such that $\int_{\R^d} \mu(x) \, \dx = 1$.
For brevity, we hide the domain of integration and write $\int$ to denote $\int_{\R^d}$, unless otherwise specified.
We also write $\mu(x) \, \dx$ and $\mu(\dx)$ interchangeably.

\subsection{Divergences between probability distributions}

Let $\rho, \pi$ be two probability distributions on $\R^d$ with $\rho \ll \pi \ll \dx$.
We recall the relative entropy or KL (Kullback-Leibler) divergence between them is:
$$\KL(\rho \,\|\, \pi) = \E_\rho\left[\log \frac{\rho}{\pi}\right] = \int_{\R^d} \rho(x) \log \frac{\rho(x)}{\pi(x)} \, \dx.$$
If $\rho$ and $\pi$ both have differentiable density functions, then we recall the relative Fisher information between $\rho$ and $\pi$ is:
$$\FI(\rho \,\|\, \pi) = \E_\rho\left[\left\|\nabla \log \frac{\rho}{\pi} \right\|^2\right] = \int_{\R^d} \rho(x) \left\|\nabla \log \frac{\rho(x)}{\pi(x)}\right\|^2 \, \dx.$$
We recall the total variation (TV) distance between $\rho$ and $\pi$ is:
$$\TV(\rho,\pi) = \sup_{A \subseteq \R^d} |\rho(A) - \pi(A)| = \frac{1}{2} \int_{\R^d} |\rho(x) - \pi(x)| \, \dx.$$
We recall the chi-square divergence between $\rho$ and $\pi$ is:
$$\chi^2(\rho\,\|\,\pi) = \Var_\pi\left(\frac{\rho}{\pi}\right) = \int_{\R^d} \left(\frac{\rho(x)}{\pi(x)}-1\right)^2 \pi(x) \, \dx = \int_{\R^d} \frac{\rho(x)^2}{\pi(x)} \, \dx - 1.$$
We recall the R\'enyi divergence of order $q > 0, q \neq 1$ between $\rho$ and $\pi$ is:
$$\Ren_q(\rho\,\|\,\pi) = \frac{1}{q-1} \log \E_\pi\left[\left(\frac{\rho}{\pi}\right)^q\right] = \frac{1}{q-1} \log \int_{\R^d} \frac{\rho(x)^q}{\pi(x)^{q-1}} \, \dx.$$
We recall that $\KL(\rho\,\|\,\pi) = \lim_{q \to 1} \Ren_q(\rho\,\|\,\pi)$.
Furthermore, we note that $\Ren_2(\rho\,\|\,\pi) = \log(1+\chi^2(\rho\,\|\,\pi))$, or equivalently, $\chi^2(\rho\,\|\,\pi) = \exp(\Ren_2(\rho\,\|\,\pi))-1$.

We recall the following relations for any $\rho \ll \pi$, that
$2\TV(\rho,\pi)^2 \le \KL(\rho\,\|\,\pi) \le \chi^2(\rho\,\|\,\pi)$; the first inequality between KL divergence and TV distance is Pinsker's inequality.
Furthermore, R\'enyi divergence is increasing in the index, so for any $1 < q \le q'$, $\KL(\rho\,\|\,\pi) \le \Ren_q(\rho\,\|\,\pi) \le \Ren_{q'}(\rho\,\|\,\pi)$.

\subsubsection{Basic properties of divergences}

All the divergences above (KL, chi-square, and R\'enyi) are not symmetric, and they do not satisfy the triangle inequality.
However, the R\'enyi divergence satisfies a generalized triangle inequality.
For any probability distributions $\mu, \rho, \pi$ with $\mu \ll \rho \ll \pi$, and for any $q > 1$ and $\lambda \in (0,1)$, we have:
$$\Ren_q(\mu \,\|\, \pi) \le \frac{q-\lambda}{q-1} \, \Ren_{q/\lambda}(\mu \,\|\, \rho) + \Ren_{(q-\lambda)/(1-\lambda)}(\rho\,\|\, \pi)\,.$$
We apply this in the proof of Lemma~\ref{Lem:FI_from_chisq_prior} with $q=2$ and $\lambda = \frac{2}{3}$.

We recall that relative Fisher information is convex (with respect to the usual linear structure of probability distributions) in the first argument.
For any probability distributions $\mu, \rho, \pi$ and any $\lambda \in [0,1]$, we have:
$$\FI(\lambda \mu + (1-\lambda) \rho \,\|\, \pi) \le \lambda \FI(\mu \,\|\, \pi) + (1-\lambda) \FI(\rho \,\|\, \pi)\,.$$

We recall the chain rule for KL divergence.
For any joint probability distributions $\rho^{XY}$ and $\pi^{XY}$, with $X$-marginals $\rho^X$ and $\pi^X$, and conditional distributions $\rho^{Y \mid X}$ and $\pi^{Y \mid X}$, respectively, we have:
$$\KL(\rho^{XY} \,\|\, \pi^{XY}) = \KL(\rho^X \,\|\, \pi^X) + \E_{x \sim \rho^X}\left[\KL(\rho^{Y \mid X=x} \,\|\, \pi^{Y \mid X=x}) \right].$$

We also recall the Donsker{--}Varadhan variational formula for KL divergence:
For any probability distributions $\rho, \pi$, and any function $g \colon \R^d \to \R$, we have:
$$\KL(\rho \,\|\, \pi) \ge \E_\rho[g] - \log \E_\pi[\exp(g)]\,.$$

\subsubsection{Channel and data processing inequality}

In this paper, we use the term \textit{channel} from information theory to describe a process that adds randomness.
Formally, a \textit{channel} $P$ from $\R^d$ to $\R^d$ is a collection of probability distributions, $P = (P_y)_{y \in \R^d}$, where $P_y$ is the output distribution of the channel when the input is a point mass at $y \in \R^d$.
Given an input probability distribution $\rho$ over $\R^d$, we denote the output of the channel by $\rho P$, given by:
$$\rho P = \int_{\R^d} \rho(y) \, P_y \, \dy.$$

A fundamental property in information theory is the \textit{data processing inequality (DPI)}, which states that statistical divergences cannot increase along a channel.
Namely, given any channel $P$, any input distributions $\rho, \pi$, and statistical divergence $D$ (including the TV distance, KL divergence, chi-square divergence, and R\'enyi divergence), we have:
$$D(\rho P \,\|\, \pi P) \le D(\rho \,\|\, \pi).$$
Notably, the data processing inequality does \textit{not} hold for the relative Fisher information; see~\cite[Example~3]{wibisono2025mixing} for a counterexample even for the Gaussian channel.

\subsection{Properties of probability distributions}

We say a probability distribution $\pi \propto \exp(-f)$ on $\R^d$ is \textit{$L$-log-smooth} for some $L \in (0,\infty)$ if $f \colon \R^d \to \R$ is twice-differentiable and $\|\nabla^2 f(x)\|_{\op} \le L$ for all $x \in \R^d$.
We say $\pi \propto \exp(-f)$ is \textit{$\alpha$-strongly log concave (SLC)} for some $\alpha > 0$ if $f \colon \R^d \to \R$ is an $\alpha$-strongly convex function.
If $\alpha = 0$, then we say $\pi$ is log-concave.

We say $\pi$ satisfies a \textit{log-Sobolev inequality (LSI)} with constant $\CLSI \in (0,\infty)$ if for any probability distribution $\rho \ll \pi$, we have 
$$\FI(\rho \,\|\, \pi) \ge \frac{2}{\CLSI} \KL(\rho \,\|\, \pi).$$
We recall that if $\pi$ is $\alpha$-SLC, then $\pi$ satisfies a LSI with constant $\CLSI = 1/\alpha$.

We say that a probability distribution $\pi$ on $\R^d$ is \textit{$\beta$-sub-Gaussian} for some $\beta \in (0, \infty)$ if for any $v \in \R^d$, 
$$\E_{X \sim \pi}[\exp(\langle v, X - \E_\pi[X] \rangle)] \le \exp\left(\frac{1}{2} \beta^2 \|v\|^2\right).$$
For a function $G \colon \R^d \to \R^m$ for $m \ge 1$, we say ``$G$ is $\beta$-sub-Gaussian under $\pi$'' to mean that the law of $G(X)$ is $\beta$-sub-Gaussian when $X \sim \pi$.
We recall that if $\pi$ is $L$-log-smooth, so $\nabla \log \pi$ is $L$-Lipschitz, then $\nabla \log \pi$ is $\sqrt{L}$-sub-Gaussian under $\pi$; see~\cite[Lemma~6]{Che+25DDPM}.

\section{Helper results}

\subsection{Sub-Gaussian score lemma}

\begin{lemma}\label{Lem:SubgScore1}
    Assume that $\pi^X \propto \exp(-f)$ is $L$-log-smooth for some $L \in (0,\infty)$.
    \begin{enumerate}
        \item For each $y\in\R^d$, $\nabla\log \RGO_y$ is $\sqrt{L+1/h}$-sub-Gaussian under $\RGO_y$.
        \item 
    $\|\nabla \log(\rho^Y \RGO/\pi^X)\|$ is $\sigma$-sub-Gaussian with respect to $\rho^Y \RGO$, where
    \begin{align*}
        \sigma^2 \lesssim \left(L+\frac{1}{h}\right)\,d + L\log(1+\chi^2(\rho^Y\mmid \pi^Y))\,.
    \end{align*}
    \end{enumerate}
\end{lemma}
\begin{proof}
The first statement follows from~\cite[Lemma 4]{Che+25DDPM}.

For the second statement, by the triangle inequality for the sub-Gaussian norm, it suffices to consider $\nabla \log \rho^Y \RGO$ and $\nabla \log \pi^X = \nabla f$ separately.
For the first term, by~\cite[Lemma 6]{Che+25DDPM}, $\nabla \log \rho^Y \RGO$ is sub-Gaussian under $\rho^Y \RGO$ if each $\nabla \log \RGO_y$ is sub-Gaussian under $\RGO_y$.
Hence, using the first statement, $\nabla \log \rho^Y \RGO$ is $\sqrt{L+1/h}$-sub-Gaussian under $\rho^Y \RGO$, and by a standard covering argument, $\|\nabla \log \rho^Y \RGO\|$ is $O(\sqrt{(L+1/h)\,d})$-sub-Gaussian under $\rho^Y \RGO$.

For the second term, the same lemmas show that $\nabla f$ is $\sqrt L$-sub-Gaussian under $\pi^X$, so there is a universal constant $c > 0$ such that for any unit vector $v\in\R^d$,
\begin{align*}
    \int \exp\Bigl( \frac{c\,\langle v,\nabla f\rangle^2}{L}\Bigr)\,\d\pi^X \le 2\,.
\end{align*}
If we know that $\bar\chi^2 \deq \chi^2(\rho^Y \RGO \mmid \pi^X) \le \chi^2(\rho^Y \mmid \pi^Y) < \infty$, then by change of measure,
\begin{align*}
    \int \exp\Bigl(\frac{(c/2)\,\langle v,\nabla f\rangle^2}{L} \Bigr)\,\d\rho^Y \RGO
    &\le 2 + \sqrt{2\chi^2(\rho^Y \RGO \mmid \pi^X)}
    = 2 + \sqrt{2\bar\chi^2}
\end{align*}
and this leads to the following tail bound: with probability at least $1-\delta$ under $\rho^Y \RGO$, $|\langle v,\nabla f\rangle| \lesssim \sqrt{L\,(\log(1+\bar\chi^2) + \log(1/\delta))}$.
By taking a union bound over a covering of the unit sphere, it implies that $\|\nabla f\|$ is $O(\sqrt{L\,(d+\log(1+\bar\chi^2))})$-sub-Gaussian under $\rho^Y \RGO$.
\end{proof}

\subsection{R\'enyi divergence along the heat flow}

\begin{lemma}\label{lem:renyi_after_heat}
    Let $\rho \propto \exp(-V)$ be a log-concave and $L$-log-smooth probability distribution on $\R^d$.
    For $t \ge 0$, let $\Q^t$ denote the Gaussian channel with variance $t$, so $\rho \Q^t = \rho \ast \N(0, tI)$.
    Let $q \ge 2$ be an integer.
    Then, for all $0 \le t < 1/(Lq)$,
    \begin{align*}
        \Ren_q(\rho \Q^t \mmid \rho) \le \frac{d}{2\,(q-1)} \log \frac{1}{1-Lqt}\,.
    \end{align*}
\end{lemma}
\begin{proof}
    Note that we can write the density of $\rho \Q^t$ as $\rho \Q^t(x) = \E[\rho(x+\xi)]$ where $\xi \sim \mathcal N(0, tI)$.
    If $\xi_1,\dotsc,\xi_q \sim \mathcal N(0, tI)$ are i.i.d., then by the convexity of $V$ and sub-Gaussianity of the score,
    \begin{align*}
        \int \left(\frac{\d(\rho \Q^t)}{\d\rho}\right)^q\,\d\rho
        &= \int \E\left[\exp\left( -\sum_{i=1}^q V(x+\xi_i) + qV(x)\right) \right]\, \rho(\dx) \\
        &\le \int \E \left[\exp\left( - \sum_{i=1}^q \langle \nabla V(x), \xi_i\rangle\right) \right]\, \rho(\dx) \\
        &\le \E \left[\exp\left(\frac{L}{2}\,\left\| \sum_{i=1}^q \xi_i\right\|^2\right) \right] \\
        &= (1-Lqt)^{-d/2} \qquad\text{for}~Lqt < 1\,.
    \end{align*}
    Hence,
    \begin{align*}
        \Ren_q(\rho \Q^t\mmid \rho)
        = \frac{1}{q-1} \log \int \left(\frac{\d(\rho \Q^t)}{\d\rho}\right)^q\,\d\rho
        &\le \frac{d}{2\,(q-1)} \log \frac{1}{1-Lqt}\,. \qedhere
    \end{align*}
\end{proof}

\section{Proofs for the ideal proximal sampler}

\subsection{Proof of Lemma~\ref{Lem:ProxExact}}
\label{Sec:LemProxExactProof}

\begin{proof}[Proof of Lemma~\ref{Lem:ProxExact}]
We prove the estimate along the forward step, then the backward step.
\begin{enumerate}
    \item \textbf{Forward step:}
    We interpolate from $\rho^X$ to $\rho^Y$ and from $\pi^X$ to $\pi^Y$ via the heat equation.
    Concretely, we define $(\rho_t)_{0 \le t \le h}$ and $(\pi_t)_{0 \le t \le h}$ by $\rho_t \deq  \rho^X \ast \N(0,tI)$ and $\pi_t \deq  \pi^X \ast \N(0,tI)$, so $\rho_0 = \rho^X$, $\pi_0 = \pi^X$, $\rho_h = \rho^Y$, $\pi_h = \pi^Y$, and $\rho_t$, $\pi_t$ evolve following the heat equation:
    \begin{align*}
        \partial_t \rho_t = \frac{1}{2}\, \Delta \rho_t\,,
        \qquad\qquad
        \partial_t \pi_t = \frac{1}{2}\, \Delta \pi_t\,.
    \end{align*}
    Then, by a standard calculation (see, e.g.,~\cite[Lemma~12]{CCSW22improved}), we have:
    \begin{align*}
        \ddt\, \KL(\rho_t \,\|\, \pi_t) = -\frac{1}{2}\, \FI(\rho_t \,\|\, \pi_t)
    \end{align*}
    so that, by integrating over $0 \le t \le h$,
    \begin{align}\label{Eq:Calc1}
        \KL(\rho_k^X \,\|\, \pi^X) - \KL(\rho_k^Y \,\|\, \pi^Y) = \frac{1}{2} \int_0^h  \FI(\rho_t \,\|\, \pi_t) \, \dt\,.
    \end{align}

    Along the simultaneous heat flow, by~\cite[Lemma~2]{wibisono2025mixing}, we also have:
    \begin{align*}
        \ddt\, \FI(\rho_t \,\|\, \pi_t)
        &= -\E_{\rho_t}\left[\left\|\nabla^2 \log \frac{\rho_t}{\pi_t}\right\|^2_{\HS} \right] - 2\E_{\rho_t}\left[\left\|\nabla \log \frac{\rho_t}{\pi_t}\right\|^2_{-\nabla^2 \log \pi_t} \right] \\
        &\le - 2\E_{\rho_t}\left[\left\|\nabla \log \frac{\rho_t}{\pi_t}\right\|^2_{-\nabla^2 \log \pi_t} \right].
    \end{align*}

    Recall (see, e.g.,~\cite[Lemma~14]{wibisono2024score}) that since $\pi_0 = \pi^X$ is $L$-log-smooth by assumption, for $0 \le t < \frac{1}{L}$, $\pi_t = \pi_0 \ast \N(0, tI)$ satisfies for all $x \in \R^d$:
    $$-\nabla^2 \log \pi_t(x) \succeq -\frac{L}{1-tL}\, I\,.$$
    Plugging this in to the above, we get:
    \begin{align*}
        \ddt\, \FI(\rho_t \,\|\, \pi_t)
        &\le \frac{2L}{1-tL}\, \E_{\rho_t}\left[\left\|\nabla \log \frac{\rho_t}{\pi_t}\right\|^2 \right]
        = \frac{2L}{1-tL}\, \FI(\rho_t \,\|\, \pi_t)\,,
    \end{align*}
    or equivalently:
    \begin{align*}
        \ddt \log \FI(\rho_t \,\|\, \pi_t)
        \le \frac{2L}{1-tL} = -2\, \ddt \log (1-tL)\,.
    \end{align*}
    For each $0 \le t \le h$, by integrating the differential inequality above from $t$ to $h$, we get:
    \begin{align*}
        \log \FI(\rho_h \,\|\, \pi_h) 
        \le \log \FI(\rho_t \,\|\, \pi_t) + 2 \log \frac{(1-tL)}{(1-h L)} \,,
    \end{align*}
    or equivalently:
    \begin{align*}
        \FI(\rho_t \,\|\, \pi_t) \ge \FI(\rho_h \,\|\, \pi_h) \, \frac{(1-h L)^2}{(1-tL)^2}\,.
    \end{align*}
    Integrating this over $0 \le t \le h$, we get:
    \begin{align*}
        \int_0^h \FI(\rho_t \,\|\, \pi_t) \, \dt 
        &\ge \FI(\rho_h \,\|\, \pi_h) \, (1-h L)^2 \, \int_0^h \frac{1}{(1-tL)^2} \, \dt \\
        &= \FI(\rho_h \,\|\, \pi_h) \, (1-h L)^2 \left(\frac{1}{L\,(1-tL)}\right)\Bigg|_{t=0}^{t=h} \\
        &= \FI(\rho_h \,\|\, \pi_h) \, (1-h L)^2 \left(\frac{h}{1-h L}\right) \\
        &= \FI(\rho^Y \,\|\, \pi^Y) \, h\, (1-h L)\,.
    \end{align*}

    Plugging this in to our estimate~\eqref{Eq:Calc1}, we get the claim in~\eqref{Eq:FwdStep}:
    \begin{align*}
        \KL(\rho^X \,\|\, \pi^X) - \KL(\rho^Y \,\|\, \pi^Y) \ge \frac{h\, (1-h L)}{2} \, \FI(\rho_k^Y \,\|\, \pi^Y)\,.
    \end{align*}

    \item \textbf{Backward step:} 
    In the backward step, we view $\rho_+^X$ is the output of the reverse Gaussian channel from $\rho^Y$, and similarly, $\pi^X$ is the output of the reverse Gaussian channel from $\pi^Y$.
    Concretely, as before, define $\pi_t = \pi_0 \ast \N(0, tI)$ for $0 \le t \le h$, and define $\mu_t \deq  \pi_{h - t}$, so $\mu_0 = \pi^Y$, $\mu_h = \pi^X$, and $\mu_t$ evolves following the backward heat equation:
    $$\partial_t \mu_t = -\frac{1}{2}\, \Delta \mu_t = -\nabla \cdot (\mu_t \nabla \log \mu_t) + \frac{1}{2}\, \Delta \mu_t\,.$$
    Define the reverse Gaussian channel:
    $$\d X_t = \nabla \log \mu_t(X_t) \, \dt + \d B_t$$
    where $(B_t)_{t \ge 0}$ is the standard Brownian motion on $\R^d$.
    If $X_t \sim \bar{\rho_t}$ evolves along this channel, then its distribution $\bar{\rho_t}$ evolves following the Fokker--Planck equation:
    $$\partial_t \bar{\rho_t} = -\nabla \cdot (\bar{\rho_t} \nabla \log \mu_t) + \frac{1}{2}\, \Delta \bar{\rho_t}\,.$$
    As shown in~\cite[Appendix~A.1.2]{CCSW22improved}, if we start the reverse Gaussian channel from a deterministic point $X_0 = y$, then the output is precisely the RGO: $X_h \sim \RGO_y = \pi^{X \mid Y}(\cdot \mid y)$.
    Therefore, if we start the reverse Gaussian channel from $X_0 \sim \bar{\rho_0} = \rho^Y$, then the output is $\bar{\rho_h} = \rho^Y \RGO = \rho^X_+$.
    
    By a standard calculation (see, e.g.,~\cite[Lemma~15]{CCSW22improved}), for $\bar{\rho_t}, \mu_t$ which evolve following the Fokker--Planck equation for the reverse Gaussian channel, we have:
    \begin{align*}
        \ddt\, \KL(\bar{\rho_t} \,\|\, \mu_t) = -\frac{1}{2}\, \FI(\bar{\rho_t} \,\|\, \mu_t)
    \end{align*}
    so that, by integrating over $0 \le t \le h$,
    \begin{align}\label{Eq:Calc2}
        \KL(\rho^Y \,\|\, \pi^Y) - \KL(\rho_{+}^X \,\|\, \pi^X) = \frac{1}{2} \int_0^h  \FI(\bar{\rho_t} \,\|\, \mu_t) \, \dt\,.
    \end{align}        

    By~\cite[Lemma~2]{wibisono2025mixing}, along the reverse Gaussian channel, we also have the formula:
    $$\ddt\, \FI(\bar{\rho_t} \,\|\, \mu_t) = -\E_{\bar{\rho_t}} \left[\left\| \nabla^2 \log \frac{\bar{\rho_t}}{\mu_t} \right\|^2_{\HS} \right] \le 0\,.$$
    Therefore, for any $0 \le t \le h$,
    \begin{align}\label{Eq:BdFIBwdCalc}
        \FI(\bar{\rho_t} \,\|\, \mu_t) \ge \FI(\bar{\rho_h} \,\|\, \mu_h)\,.
    \end{align}
    In particular, when $t=0$, 
    $$\FI(\rho^Y \,\|\, \pi^Y) = \FI(\bar{\rho_0} \,\|\, \mu_0) \ge \FI(\bar{\rho_h} \,\|\, \mu_h) = \FI(\rho_{+}^X \,\|\, \pi^X)\,,$$
    as claimed in the first case of~\eqref{Eq:BwdStep}.
    
    Plugging in~\eqref{Eq:BdFIBwdCalc} to~\eqref{Eq:Calc2}, we also get the second case in the claim~\eqref{Eq:BwdStep}:
    \begin{align*}
        \KL(\rho^Y \,\|\, \pi^Y) - \KL(\rho_{+}^X \,\|\, \pi^X) = \frac{1}{2} \int_0^h  \FI(\bar{\rho_t} \,\|\, \mu_t) \, \dt
        \ge \frac{h}{2} \, \FI(\bar\rho_h \,\|\, \mu_h)
        = \frac{h}{2} \, \FI(\rho_{+}^X \,\|\, \pi^X)\,.
    \end{align*}      
\end{enumerate}

\end{proof}

\subsection{Proof of Theorem~\ref{Thm:FIExactRGO}}
\label{Sec:FIExactRGOProof}

\begin{proof}[Proof of Theorem~\ref{Thm:FIExactRGO}]
    The implication~\eqref{Eq:BoundAvg} follows from the bound~\eqref{Eq:BoundTotal} by the convexity of relative Fisher information in the first argument,
    so it suffices to prove the bound~\eqref{Eq:BoundTotal}.
    We analyze each iteration of the proximal sampler and invoke the bounds from Lemma~\ref{Lem:ProxExact}:
    \begin{enumerate}
        \item In the forward step, let $\rho_k^Y = \rho_k^X \Q^h = \rho_k^X \ast \N(0, h I)$ and recall $\pi^Y = \pi^X \Q^h = \pi^X \ast \N(0, h I)$.
        By the bound~\eqref{Eq:FwdStep} and the first case of the bound~\eqref{Eq:BwdStep} from Lemma~\ref{Lem:ProxExact}, we have:
        \begin{align*}
            \KL(\rho_k^X \,\|\, \pi^X) - \KL(\rho_k^Y \,\|\, \pi^Y) \ge \frac{h\, (1-h L)}{2} \, \FI(\rho_k^Y \,\|\, \pi^Y) \ge \frac{h\, (1-h L)}{2} \, \FI(\rho_{k+1}^X \,\|\, \pi^X)\,.
        \end{align*}

        \item In the backward step, by the second case of the bound~\eqref{Eq:BwdStep} from Lemma~\ref{Lem:ProxExact}, we have
        \begin{align*}
            \KL(\rho_k^Y \,\|\, \pi^Y) - \KL(\rho_{k+1}^X \,\|\, \pi^X) \ge \frac{h}{2} \, \FI(\rho_{k+1}^X \,\|\, \pi^X)\,.
        \end{align*}   
    \end{enumerate}

    Combining the forward and backward steps, we obtain:
    \begin{align*}
        \KL(\rho_k^X \,\|\, \pi^X) - \KL(\rho_{k+1}^X \,\|\, \pi^X) 
        &= 
        \KL(\rho_k^X \,\|\, \pi^X) - \KL(\rho_k^Y \,\|\, \pi^Y) + \KL(\rho_k^Y \,\|\, \pi^Y) - \KL(\rho_{k+1}^X \,\|\, \pi^X) \\
        &\ge \frac{h\, (1-h L)}{2} \, \FI(\rho_{k+1}^X \,\|\, \pi^X) + \frac{h}{2} \, \FI(\rho_{k+1}^X \,\|\, \pi^X) \\
        &= h\, \Bigl(1-\frac{h L}{2}\Bigr) \, \FI(\rho_{k+1}^X \,\|\, \pi^X) \,.
    \end{align*}
    Summing the bound above for $0 \le k \le K-1$, and dropping the last term $\KL(\rho_K^X \,\|\, \pi^X)$, we obtain the claimed bound~\eqref{Eq:BoundTotal}.  
\end{proof}

\section{Proofs for the proximal sampler with an approximate RGO implementation}

\subsection{Proof of Lemma~\ref{Lem:FIApproxRGO}}
\label{Sec:LemFIApproxRGOProof}

\begin{proof}[Proof of Lemma~\ref{Lem:FIApproxRGO}]
    Let $\widetilde \rho_+^X = \rho^Y \RGO$
    denote the application of the exact RGO channel to $\rho^Y = \rho^X \Q^h$.
    
    First, we observe the following sequence of inequalities:
    \begin{align}
        \frac{h\,(1-Lh)}{4}\,\FI(\rho_+^X\,\|\,\pi^X)
        &\stackrel{(i)}{\le} \frac{h\,(1-Lh)}{2}\,\FI(\widetilde\rho_+^X\,\|\,\pi^X) + \frac{h\,(1-Lh)}{4}\,\bar\err_\FI^2  \notag \\
        &\stackrel{(ii)}{\le} \frac{h\,(1-Lh)}{2}\,\FI(\rho^Y\,\|\,\pi^Y) + \frac{h\,(1-Lh)}{4}\,\bar\err_\FI^2  \notag \\
        &\stackrel{(iii)}{\le} \KL(\rho^X \,\|\, \pi^X) - \KL(\rho^Y \,\|\, \pi^Y) + \frac{h\,(1-Lh)}{4}\, \bar\err_\FI^2\,. \label{Eq:FIApproxRGO-Calc1}
    \end{align}
    In the above, the first step $(i)$ is from the error guarantee~\eqref{Eq:ApproxGuaranteeFI} in relative Fisher information for the approximate RGO;
    the second step $(ii)$ is from the first case in the bound~\eqref{Eq:BwdStep} in Lemma~\ref{Lem:ProxExact}, which states that relative Fisher information is decreasing along the exact RGO channel; and the third step $(iii)$ is
    from the bound~\eqref{Eq:FwdStep} in Lemma~\ref{Lem:ProxExact} along the forward step of the proximal sampler.

    Next, we also observe the following sequence of inequalities:
    \begin{align}
        \frac{h}{4}\,\FI(\rho_+^X\,\|\, \pi^X)
        &\stackrel{(i)}{\le} \frac{h}{2}\,\FI(\widetilde\rho_+^X\,\|\,\pi^X) + \frac{h}{4}\,\bar\err_\FI^2  \notag \\
        &\stackrel{(ii)}{\le} \KL(\rho^Y \,\|\, \pi^Y) - \KL(\widetilde \rho_+^X \,\|\, \pi^X) +\frac{h}{4}\,\bar\err_\FI^2  \notag \\
        &\stackrel{(iii)}{\le} \KL(\rho^Y \,\|\, \pi^Y) - \KL(\rho_+^X \,\|\, \pi^X) + \bar\err_\KL^2 + \frac{h}{4}\,\bar\err_\FI^2\,. \label{Eq:FIApproxRGO-Calc2}
    \end{align}
    In the above, the first step $(i)$ is again from the error guarantee~\eqref{Eq:ApproxGuaranteeFI} in relative Fisher information for the approximate RGO;
    the second step $(ii)$ is from the second case in the bound~\eqref{Eq:BwdStep} in Lemma~\ref{Lem:ProxExact}, which bounds relative Fisher information in terms of the change in KL divergence along the exact RGO channel;
    and the third step $(iii)$ is from the error guarantee~\eqref{Eq:ApproxGuaranteeKL} in KL divergence for the approximate RGO.

    Adding the two estimates~\eqref{Eq:FIApproxRGO-Calc1} and~\eqref{Eq:FIApproxRGO-Calc2} above yields the desired bound.
\end{proof}

\subsection{Proof of Lemma~\ref{lem:KL_bd}}
\label{Sec:LemKL_bdProof}

\begin{proof}[Proof of Lemma~\ref{lem:KL_bd}]
    We write the following decomposition:
    \begin{align}
        \KL(\mu \mmid \pi)
        &= \int \log\frac{\d\mu}{\d\pi}\,\d\mu \notag \\
        &= \int \log\frac{\d\mu}{\d\nu}\,\d\mu + \int \log\frac{\d\nu}{\d\pi}\,\d\nu 
        + \int \log\frac{\d\nu}{\d\pi}\,\d(\mu - \nu) \notag \\
        &= \KL(\mu\mmid \nu) + \KL(\nu\mmid \pi) + \int \left(\log \frac{\d\nu}{\d\pi} \right) \cdot \left(\frac{\d\mu}{\d\nu}-1\right) \, \d\nu \notag \\
        &\le \KL(\mu\mmid \nu) + \KL(\nu\mmid \pi) + \sqrt{\E_\nu\left[\left(\log \frac{\d\nu}{\d\pi}\right)^2 \right] \,\chi^2(\mu\mmid\nu)}\,, \label{Eq:KL_bd_Calc}
    \end{align}
    where in the last step above, we have applied Cauchy--Schwarz inequality.

    Let $\ell \deq \log\frac{\d\nu}{\d\pi}$ and $R \deq \Ren_2(\nu\mmid\pi)$.
    Then, we know that $\E_\nu[\exp(\ell)] = 1+\chi^2(\nu\mmid\pi) = \exp(R)$, and hence Markov's inequality yields 
    $$\nu(\ell \ge t) = \nu(\exp(\ell) \ge \exp(t)) \le \exp(R-t)$$ 
    for $t\ge 0$.
    We also have $\E_\nu[\exp(-\ell)] \le 1$ (with equality when $\pi \ll \nu$), which yields 
    $$\nu(\ell \le -t) = \nu(\exp(-\ell) \ge \exp(t)) \le \exp(-t)$$ 
    for $t \ge 0$.
    Integrating these tail bounds yields the following, for $\ell_+ = \max\{\ell, 0\}$ and $\ell_- = -\min\{\ell, 0\}$:
    \begin{align*}
        \E_\nu\left[(\ell_+)^2\right]
        &= 2\int_0^\infty t \, \nu(\ell\ge t)\,\dt
        \le 2\int_0^R t\,\dt + 2\int_R^\infty t\exp(R-t)\,\dt
        = R^2 + 2R + 2\,, \\
        \E_\nu\left[(\ell_-)^2\right]
        &= 2\int_0^\infty t \, \nu(\ell\le -t)\,\dt
        \le 2\int_0^\infty t\exp(-t)\,\dt
        = 2\,.
    \end{align*}
    In the above, we have used the standard calculation (via integration by parts) $\int_0^\infty t \exp(-t) \, \dt = 1$, and
    $\int_R^\infty t\exp(R-t)\,\dt
    = R \int_R^\infty \exp(R-t)\,\dt + \int_R^\infty (t-R) \exp(R-t)\,\dt = R + 1.$
    
    Hence, since $\ell^2 = (\ell_+)^2 + (\ell_-)^2$, we have 
    $$\E_\nu\left[\ell^2\right] 
    = \E_\nu\left[(\ell_+)^2\right] + \E_\nu\left[(\ell_-)^2\right] 
    \le R^2 + 2R + 4
    \le (R+2)^2 \,.$$
    Plugging in this bound to the earlier calculation~\eqref{Eq:KL_bd_Calc} completes the proof.
\end{proof}

\subsection{Proof of Lemma~\ref{lem:FI_bd}}
\label{Sec:LemFI_bdProof}

\begin{proof}[Proof of Lemma~\ref{lem:FI_bd}]
Let $\widehat X \sim \rho^Y \hRGO$ and $X \sim \rho^Y \RGO$.
Then,
\begin{align*}
    \FI(\rho^Y\hRGO \mmid \pi^X) & - 2\,\FI(\rho^Y \RGO \mmid \pi^X) \\
    &= \left\| \nabla \log \frac{\rho^Y\hRGO}{\pi^X}(\widehat X)\right\|_{L^2}^2 - 2\,\left\| \nabla \log \frac{\rho^Y \RGO}{\pi^X}(X)\right\|_{L^2}^2 \\
    &\le 2\,\underbrace{\left\| \nabla \log \frac{\rho^Y \hRGO}{\rho^Y \RGO}(\widehat X)\right\|_{L^2}^2}_{\eqqcolon (\msf{I})} 
    + 2\,\underbrace{\left(\left\| \nabla \log \frac{\rho^Y \RGO}{\pi^X}(\widehat X)\right\|_{L^2}^2 - \left\|\nabla \log \frac{\rho^Y \RGO}{\pi^X}(X)\right\|_{L^2}^2\right)}_{\eqqcolon (\mathsf{II})}\,.
\end{align*}

We analyze the terms above.
\begin{enumerate}
    \item \textbf{Last two terms.}
    Starting with the last two terms, we consider
\begin{align*}
    (\mathsf{II}) = \E\Bigg[\underbrace{\left\|\nabla \log \frac{\rho^Y\RGO}{\pi^X}(\widehat X)\right\|^2 -\E\left[\left\|\nabla \log \frac{\rho^Y\RGO}{\pi^X}(X)\right\|^2\right]}_{\eqqcolon \zeta(\widehat X)}\Bigg] = \E\left[\zeta(\widehat X)\right]\,.
\end{align*}
By Donsker{--}Varadhan, for any $\lambda > 0$,
\begin{align*}
    \E\left[\zeta(\widehat X)\right]
    &\le \frac{1}{\lambda}\,\bigl\{\KL(\rho^Y\hRGO \mmid \rho^Y \RGO) + \log \E\left[\exp(\lambda\zeta(X))\right]\bigr\}\,.
\end{align*}
By Lemma~\ref{Lem:SubgScore1},
$\|\nabla \log(\rho^Y \RGO/\pi^X)\|$ is $\sigma$-sub-Gaussian with respect to $\rho^Y \RGO$, where $\sigma \lesssim \left(L+\frac{1}{h}\right)\,d + L\log(1+\chi^2(\rho^Y\mmid \pi^Y))\,.$
Since $\zeta(X)$ is centered, standard results on sub-Gaussianity show that
\begin{align*}
    \log \E \left[\exp(\lambda \zeta(X)) \right]
    &\lesssim \lambda^2 \sigma^4 \qquad\text{for all}~0 \le \lambda \lesssim \frac{1}{\sigma^2}\,.
\end{align*}
Plugging this in above and choosing the optimal $\lambda = \sigma^{-2} \sqrt{\KL(\rho^Y\hRGO \mmid \rho^Y \RGO)}$ (which is a valid choice $\lambda \lesssim \sigma^{-2}$ provided that $\err_{\KL} \ll 1$) yields:
\begin{align*}
    \E\left[\zeta(\widehat X)\right]
    &\lesssim \frac{\KL(\rho^Y\hRGO \mmid \rho^Y \RGO)}{\lambda} + \lambda \sigma^4 
    \asymp \sigma^2\sqrt{\KL(\rho^Y\hRGO \mmid \rho^Y \RGO)} \le \sigma^2 \err_{\KL}\,.
\end{align*}

    \item 
\textbf{First term.}
The first term above is $(\msf{I}) = \FI(\rho^Y \hRGO \mmid \rho^Y \RGO)$.
Note that since
\begin{align*}
    \rho^Y \RGO(x)
    &= \int \RGO_y(x)\, \rho^Y(\dy)\,,
\end{align*}
we have
\begin{align*}
    \nabla \log \rho^Y \RGO(x)
    &= \int \nabla \log \RGO_y(x)\, \frac{\RGO_y(x)\,\rho^Y(y)}{\int \RGO_{y'}(x)\,\rho^Y(\d y')}\,\dy\,,
\end{align*}
and similarly for $\nabla \log \rho^Y \hRGO(x)$.
Therefore,
\begin{align*}
    &\left\| \nabla \log \frac{\rho^Y \hRGO}{\rho^Y \RGO} \right\|_{L^2(\rho^Y \hRGO)}^2
    \le 2 \iint \left\| \nabla \log \frac{\hRGO_y}{\RGO_y}(x)\right\|^2 \,\frac{\hRGO_y(x)\,\rho^Y(y)}{\int \hRGO_{y'}(x)\,\rho^Y(\d  y')}\,\dy\,\rho^Y\hRGO(\dx) \\
    &\qquad{} + 2\int\left\| \int \nabla \log \RGO_y(x)\,\left(\frac{\hRGO_y(x)\,\rho^Y(y)}{\int \hRGO_{y'}(x)\,\rho^Y(\d y')} - \frac{\RGO_y(x)\,\rho^Y(y)}{\int \RGO_{y'}(x)\,\rho^Y(\d y')}\right)\,\dy \right\|^2\,\rho^Y\hRGO(\dx)\,.
\end{align*}
The first sub-term above is
\begin{align*}
    2 \iint \left\| \nabla \log \frac{\hRGO_y}{\RGO_y}(x)\right\|^2 \,\hRGO_y(\dx)\,\rho^Y(\dy)
    = 2 \int \FI(\hRGO_y \mmid \RGO_y)\,\rho^Y(\dy)
\end{align*}
so it is controlled provided our implementation satisfies $\sup_{y\in\R^d} \FI(\hRGO_y \mmid \RGO_y) \le \err_{\FI}^2$.

The second sub-term above is
\begin{align*}
    \int \left\| \int \nabla \log \RGO_y(x)\,\left(\frac{\hRGO_y(x)\,\rho^Y(y)}{\int \hRGO_{y'}(x)\,\rho^Y(\d y')} - \frac{\RGO_y(x)\,\rho^Y(y)}{\int \RGO_{y'}(x)\,\rho^Y(\d y')}\right)\,\dy \right\|^2 \,\rho^Y\hRGO(\dx) \,.
\end{align*}
Let $\widehat \rho^{Y\mid X=x}$, $\rho^{Y\mid X=x}$ be the two probability measures appearing in the inner integral above.
Also, we decompose the score as
\begin{align*}
    S_1(x,y) = \nabla \log \RGO_y(x) \one_{\|\nabla \log \RGO_y(x)\| \le B}\,, \qquad S_2(x,y) = \nabla \log \RGO_y(x) \one_{\|\nabla \log \RGO_y(x)\| > B}
\end{align*}
for some threshold $B$ that we will choose below.
We handle these two parts separately. 
\begin{enumerate}
    \item First, for $S_1$:
\begin{align*}
    &\int \left\| \int S_1(x,y)\,\widehat \rho^{Y\mid X=x}(\dy) - \int S_1(x,y) \,\rho^{Y\mid X=x}(\dy) \right\|^2 \,\rho^Y\hRGO(\dx) \\
    &\qquad\lesssim B^2 \int \TV(\widehat \rho^{Y\mid X=x}, \rho^{Y\mid X=x})^2\,\rho^Y\hRGO(\dx) \\
    &\qquad\lesssim B^2 \int \KL(\widehat \rho^{Y\mid X=x} \mmid \rho^{Y\mid X=x})\,\rho^Y\hRGO(\dx) \\
    &\qquad \le B^2 \err_{\KL}^2\,,
\end{align*}
by Pinsker's inequality and the KL chain rule.

    \item Now for the unbounded part $S_2$: by Donsker{--}Varadhan, for any $\lambda > 0$,
\begin{align*}
    &\int \left\| \int \left(S_2(x,y) - \int S_2(x,y') \,\rho^{Y\mid X=x}(\d y')\right)\, \widehat \rho^{Y\mid X=x}(\dy) \right\|^2 \,\rho^Y\hRGO(\dx) \\
    &\qquad \le \iint \left\| S_2(x,y) - \int S_2(x,y') \,\rho^{Y\mid X=x}(\d y')\right\|^2\, \widehat \rho^{Y\mid X=x}(\dy) \,\rho^Y\hRGO(\dx) \\
    &\qquad \le \frac{1}{\lambda}\, \KL(\widehat\rho^{X,Y} \mmid \rho^{X,Y}) + \frac{1}{\lambda} \log \E_{\rho^{X,Y}} \left[\exp(\lambda\,\|S_2(X,Y) - \E[S_2(X,Y)\mid X]\|^2) \right]\,.
\end{align*}
Here, $\rho^{X,Y}$ is the joint distribution with $Y$-marginal $\rho^Y$ and conditional $\RGO$; and similarly $\widehat\rho^{X,Y}$ is the joint distribution with $Y$-marginal $\rho^Y$ and conditional $\hRGO$.
By the KL chain rule, one sees that $\KL(\widehat\rho^{X,Y} \mmid \rho^{X,Y}) \le \err_{\KL}^2$.
We want to check that $S_2(X,Y)$ is sub-Gaussian.
By Jensen's inequality, if $Y' \sim \rho^{Y\mid X}$ is conditionally independent of $Y$ given $X$,
\begin{align*}
    \E\left[\exp(\lambda\,\|S_2(X,Y) - \E[S_2(X,Y) \mid X]\|^2) \right] 
    &\le \E\left[\exp(\lambda \,\|S_2(X,Y) - S_2(X,Y')\|^2) \right] \\
    &\le \E\left[\exp(2\lambda\,(\|S_2(X,Y)\|^2 + \|S_2(X,Y')\|^2)) \right] \\
    &\le \E \left[\exp(4\lambda\,\|S_2(X,Y)\|^2) \right]\,,
\end{align*}
where the last line is by Cauchy--Schwarz.
Note that
\begin{align*}
    \E\left[\exp(4\lambda\,\|S_2(X,Y)\|^2)\right]
    &\le 1+ \E\left[\exp(4\lambda\,\|\nabla \log\RGO_Y(X)\|^2) \, \one_{\|\nabla \log \RGO_Y(X)\|\ge B}\right] \\
    &\le 1 + \sqrt{\E\left[\exp(8\lambda\,\|\nabla \log\RGO_Y(X)\|^2) \, \Pr(\|\nabla \log \RGO_Y(X)\| \ge B)\right]}\,.
\end{align*}
On the other hand, by Lemma~\ref{Lem:SubgScore1}, the expectation of the exponential is a constant provided $\lambda^{-1} \gg (L+1/h)\,d$.
If we choose $B \asymp \sqrt{(L+1/h)\,(d+\log(1/\delta))}$, then the probability term is $O(\delta)$.
\end{enumerate}

Hence, our overall bound for $(\msf{I})$ is
\begin{align*}
    \err_{\FI}^2 + (L+1/h)\,(d+\log(1/\delta))\,\err_{\KL}^2 + (L+1/h)\,d\,(\err_{\KL}^2 + \sqrt\delta)
\end{align*}
up to a universal constant.
\end{enumerate}

\paragraph{Summary.}
Putting together our bounds, we have shown that for any $\delta > 0$,
\begin{align*}
    &\FI(\rho^Y \hRGO \mmid \pi^X) - 2\,\FI(\rho^Y \RGO\mmid \pi^X) \\
    &\qquad \lesssim(L+1/h)\,(d+\log(1+\bar\chi^2))\,\err_{\KL} + \err_{\FI}^2 \\
    &\qquad\qquad{} + (L+1/h)\,(d+\log(1/\delta))\,\err_{\KL}^2 + (L+1/h)\,d\sqrt\delta\,.
\end{align*}
Choosing $\delta = \err_{\KL}^4$ leads to the final result.
\end{proof}

\subsection{Chi-squared divergence along the iterates}
\label{Sec:Lemchisq_prox_sampler}

We also need the following lemma, which controls the chi-squared divergence along the iterates.

\begin{lemma}\label{lem:chisq_prox_sampler}
    Assume we have an approximate RGO implementation $\hRGO$ which satisfies
    \begin{align*}
        \sup_{y\in\R^d} \chi^2(\hRGO_y \mmid \RGO_y) \le \err_{\chi^2}^2\,.
    \end{align*}
    Consider the proximal sampler algorithm with approximate RGO implementation $\hRGO$.
    Namely, from $\rho_0^X$, define for all $k \ge 0$: $\rho_k^Y = \rho_k^X \Q^h$ and $\rho_{k+1}^X = \rho_k^Y \hRGO$.
    Then for all $k \ge 0$,
    \begin{align*}
        \max\{\Ren_2(\rho_k^X \mmid \pi^X),\,\Ren_2(\rho_k^Y \mmid \pi^Y)\} \le \Ren_2(\rho_0^X \mmid \pi^X) + k\log(1+\err_{\chi^2}^2)\,.
    \end{align*}
    Equivalently,
    \begin{align*}
        \max\{\chi^2(\rho_k^X \mmid \pi^X),\,\chi^2(\rho_k^Y \mmid \pi^Y)\} \le (1+\err_{\chi^2}^2)^k\,\chi^2(\rho_0^X \mmid \pi^X) + (1+\err_{\chi^2}^2)^k - 1\,.
    \end{align*}
\end{lemma}
\begin{proof}
    We follow the error analysis in the proof of~\cite[Theorem 5.4]{AltChe24Warm}.
    First, along the forward step of the proximal sampler, by the data processing inequality,
    \begin{align*}
        \Ren_2(\rho_k^Y \mmid \pi^Y)
        = \Ren_2(\rho_k^X \Q^h \mmid \pi^X \Q^h)
        &\le \Ren_2(\rho_k^X \mmid \pi^X)\,.
    \end{align*}
    Then, along the backward step with the approximate RGO, the error analysis of the proximal sampler from~\cite[Theorem 5.4]{AltChe24Warm} yields
    \begin{align*}
        \Ren_2(\rho_{k+1}^X \mmid \pi^X)
        = \Ren_2(\rho_{k}^Y \hRGO \mmid \pi^Y \RGO)
        &\le \Ren_2(\rho_k^Y \mmid \pi^Y) + \log(1+\err_{\chi^2}^2)\,.
    \end{align*}
    Iterating these inequalities yields the result.
\end{proof}

\subsection{Lemma on Fisher information bound from chi-square for approximate RGO}
\label{Sec:FI_from_chisq_prior}

To prove Corollary~\ref{cor:FI_from_chisq}, we establish the following lemma.

\begin{lemma}\label{Lem:FI_from_chisq_prior}
    Assume $\pi^X \propto \exp(-f)$ is $L$-log-smooth for some $L \in (0,\infty)$.
    Assume $h \in (0, \frac{1}{2L}]$,
    and let $\RGO \equiv \RGO^h$ denote the RGO channel with step size $h$.
    Suppose there is an approximate RGO implementation $\widetilde{\msf R}$ with the guarantees that
    \begin{align*}
        \sup_{y\in\R^d} \chi^2(\widetilde{\msf R}_y \mmid \RGO_y) \le \err_{\chi^2}^2 \le 1\,.
    \end{align*}
    Then there is an approximate implementation $\hRGO$ which satisfies
    \begin{align*}
        \frac{1}{h}\sup_{y\in\R^d} \KL(\hRGO_y \mmid \RGO_y) \le \sup_{y\in\R^d} \FI(\hRGO_y \mmid \RGO_y)
        &\lesssim \frac{d}{h}\sqrt{\err_{\chi^2}\,\left(1 + \frac{\log(1/\err_{\chi^2})}{d}\right)}\,.
    \end{align*}
    Furthermore, if the approximate implementation $\widetilde{\msf R}$ satisfies
    \begin{align*}
        \sup_{y\in\R^d}\Ren_3(\widetilde{\msf R}_y \mmid \RGO_y) \le \err_{\Ren_3}^2 \lesssim 1 \qquad\text{and}\qquad\err_{\chi^2} \ll 1/d^2\,,
    \end{align*}
    then we can also ensure that $\hRGO$ satisfies
    \begin{align*}
        \sup_{y\in\R^d} \chi^2(\hRGO_y\mmid \RGO_y) \lesssim \err_{\Ren_3}^2 + d\sqrt{\err_{\chi^2}\,\left(1 + \frac{\log(1/\err_{\chi^2})}{d}\right)}\,.
    \end{align*}
\end{lemma}
\begin{proof}
    Recall $\Q^t$ denotes the Gaussian channel with variance $t \ge 0$, so $\rho \Q^t = \rho \ast \N(0, tI)$.
    By our assumption on $h$, we know that $\RGO_y$ is $3/(2h)$-log-smooth.    
    By~\cite[Lemma 12]{Che+23FisherLower}, if we take $\hRGO_y = \widetilde{\msf R}_y \Q^t$ for some $t \ll h$, then
    \begin{align*}
        \FI(\hRGO_y \mmid \RGO_y) \lesssim \frac{\err_{\chi^2}\,(d+\log(1/\err_{\chi^2}))}{t} + \frac{dt}{h^2}\,.
    \end{align*}
    We choose $t\asymp h \sqrt{\err_{\chi^2}\,(1+\frac{1}{d}\log(1/\err_{\chi^2}))}$ to minimize this bound, resulting in the claimed bound in Fisher information.
    Furthermore, by the assumption on $h$, we also know that $\RGO_y$ is $1/(2h)$-strongly log-concave, so it satisfies the log-Sobolev inequality with constant $2h$, which implies $\frac{1}{h} \KL(\hRGO_y \mmid \RGO_y) \le \FI(\hRGO_y \mmid \RGO_y)$.
    Therefore, the claimed KL bound follows from the Fisher information bound above.
    
    For the final inequality, since $h \le \frac{1}{2L}$ and $\RGO_y$ is $3/(2h)$-log-smooth, by Lemma~\ref{lem:renyi_after_heat}, for $t < \frac{1}{4L}$ (which is satisfied by our choice $t\asymp h \sqrt{\err_{\chi^2}\,(1+\frac{1}{d} \log(1/\err_{\chi^2}))}$),
    \begin{align*}
        \Ren_4(\RGO_y \Q^t \mmid \RGO_y)
        \,\le\, \frac{d}{6} \log \frac{1}{1-\frac{2t}{h}}
        \,\lesssim\, \frac{dt}{h}
        \,\lesssim\, d\sqrt{\err_{\chi^2}\,\left(1 + \frac{\log(1/\err_{\chi^2})}{d}\right)}\,.
    \end{align*}
    Furthermore, by the data processing inequality,
    \begin{align*}
        \Ren_3(\hRGO_y \mmid \RGO_y \Q^t)
        = \Ren_3(\widetilde{\msf R}_y \Q^t \mmid \RGO_y \Q^t)
        \le \Ren_3(\widetilde{\msf R}_y \mmid \RGO_y)
        \le \err_{\Ren_3}^2\,.
    \end{align*}
    Then, by the R\'enyi triangle inequality:
    \begin{align*}
        \Ren_2(\hRGO_y \mmid \RGO_y)
        &\le \frac{4}{3}\,\Ren_3(\hRGO_y \mmid \RGO_y \Q^t) + \Ren_4(\RGO_y \Q^t \mmid \RGO_y) \\
        &\lesssim \err_{\Ren_3}^2 + d\sqrt{\err_{\chi^2}\,\left(1 + \frac{\log(1/\err_{\chi^2})}{d}\right)}\,.
    \end{align*}
    Since we assume $\err_{\chi^2} \ll 1/d^2$, the bound above implies $\Ren_2(\hRGO_y \mmid \RGO_y) \lesssim 1$, and thus we have $\chi^2(\hRGO_y\mmid \RGO_y) = \exp(\Ren_2(\hRGO_y \mmid \RGO_y)) - 1 \lesssim \Ren_2(\hRGO_y \mmid \RGO_y)$.
    This finishes the proof.
\end{proof}

\subsection{Proof of Corollary~\ref{cor:FI_from_chisq}}
\label{Sec:CorFI_from_chisq_proof}

Corollary~\ref{cor:FI_from_chisq} follows by combining Lemma~\ref{lem:FI_bd} and Lemma~\ref{Lem:FI_from_chisq_prior}. 

\begin{proof}[Proof of Corollary~\ref{cor:FI_from_chisq}]
    We use the same construction of $\hRGO$ as in Lemma~\ref{Lem:FI_from_chisq_prior}, which gives the claim for $\chi^2(\hRGO_y\mmid \RGO_y)$.
    Since we assume $\err_{\chi^2} \ll 1/d^2$, Lemma~\ref{Lem:FI_from_chisq_prior} gives the following guarantees for the approximate RGO channel $\hRGO$:
    \begin{align*}
        \err_{\KL}^2 &= \sqrt{d \err_{\chi^2}\,\left(d + \log(1/\err_{\chi^2})\right)} \ll 1\,, \\
        \err_{\FI}^2 &\lesssim L \sqrt{d \err_{\chi^2}\,\left(d + \log(1/\err_{\chi^2})\right)}\,.
    \end{align*}
    In particular, note that $\log (1/\err_{\KL}) \lesssim \log (1/\err_{\chi^2})$.
    Plugging in to the conclusion of Lemma~\ref{lem:FI_bd}, we obtain for any probability distribution $\rho^Y$:
    \begin{align*}
        \FI&(\rho^Y \hRGO \mmid \pi^X) - 2\,\FI(\rho^Y \RGO\mmid \pi^X) \\
        &\lesssim \left(L+ \frac{1}{h} \right) \left(d+\log (1/\err_{\KL}) + \Ren_2(\rho^Y \mmid \pi^Y) \right) \,\err_\KL + \err_\FI^2 \\
        &\lesssim L \left(d+\log (1/\err_{\chi^2}) + \Ren_2(\rho^Y \mmid \pi^Y) \right) \, \left\{d(d+\log(1/\err_{\chi^2}))\,\err_{\chi^2} \right\}^{1/4} + L \sqrt{d \err_{\chi^2}\,\left(d + \log(1/\err_{\chi^2})\right)} \\
        &\lesssim L \left(d+\log(1/\err_{\chi^2}) + \Ren_2(\rho^Y\mmid \pi^Y) \right)\, \left\{d(d+\log(1/\err_{\chi^2}))\,\err_{\chi^2} \right\}^{1/4} \,,
    \end{align*}
    as desired.
\end{proof}

\subsection{Proof of Theorem~\ref{thm:main}}
\label{Sec:ProofThmMain}

\begin{proof}[Proof of Theorem~\ref{thm:main}]
Set $h = \frac{1}{2L}$, and let $\RGO \equiv \RGO^h$ be the exact RGO channel.
Suppose we have an approximate RGO implementation $\hRGO \equiv \hRGO^h$, that we will construct below.
From $\rho_0^X$, we define $\rho_k^Y = \rho_k^X \Q^h$ and $\rho_{k+1}^X = \rho_k^Y \hRGO^h$, for $0 \le k \le K$ for some $K$ that we will choose below.
Suppose that $\hRGO$ has accuracies
\begin{subequations}\label{Eq:hRGO_acc}
\begin{align}
    \err_{\chi^2}^2
    &\deq \sup_{y\in\R^d} \chi^2(\hRGO_y \mmid \RGO_y)\,, \\
    \bar\err_{\KL}^2
    &\deq \max_{0 \le k \le K}\bigl\{\KL(\rho^Y\hRGO\mmid \pi^X) - \KL(\rho^Y \RGO\mmid \pi^X)\bigr\}_+\,, \\
    \bar\err_{\FI}^2
    &\deq \max_{0 \le k \le K}\bigl\{\FI(\rho_k^Y \hRGO\mmid \pi^X) - 2\,\FI(\rho_k^Y \RGO \mmid \pi^X)\bigr\}_+\,.
\end{align}
\end{subequations}
By the bound from Lemma~\ref{Lem:FIApproxRGO}, with $h = \frac{1}{2L}$, $\rho^X \leftarrow \rho_k^X$, and $\rho_+^X \leftarrow \rho_{k+1}^X$, we have for $0 \le k \le K-1$:
\begin{align*}
    \frac{3}{16L}\, \FI(\rho_{k+1}^X \,\|\, \pi^X)
    \le \KL(\rho_k^X \,\|\, \pi^X) - \KL(\rho_{k+1}^X \,\|\, \pi^X) + \bar\err_\KL^2 + \frac{3}{8L}\, \bar\err_\FI^2\,.
\end{align*}
Let $\bar\rho_K^X \deq \frac{1}{K} \sum_{k=1}^K \rho_k^X$.
Summing the above over $0 \le k \le K-1$ and using the convexity of the Fisher information,
\begin{align}\label{eq:final_fi_bd}
    \FI(\bar\rho_K^X\mmid \pi^X)
    &\le \frac{1}{K} \sum_{k=1}^K \FI(\rho_k^X\mmid \pi^X)
    \lesssim \frac{L\,\KL(\rho_0^X \mmid \pi^X)}{K} + L\bar\err_{\KL}^2 + \bar\err_{\FI}^2\,.
\end{align}

Denote by $\Delta_0 \deq \Ren_2(\rho_0^X\mmid \pi^X) = \log(1+\chi^2(\rho_0^X \mmid \pi^X))$ and $\KL_0 \deq \KL(\rho_0^X \mmid \pi^X) \le \Delta_0$.
First, we choose $K \asymp L\,\KL_0/\err^2$ to make the first term in~\eqref{eq:final_fi_bd} at most $\err^2$.
From Lemma~\ref{lem:chisq_prox_sampler},
\begin{align}\label{eq:renyi_along_path}
    \max_{0 \le k \le K}\Ren_2(\rho_k^Y \mmid \pi^Y)
    &\le \Delta_0 + K \log(1+\err_{\chi^2}^2)
    \le \Delta_0\,\left(1 + \frac{L \log(1+\err_{\chi^2}^2)}{\err^2}\right)\,.
\end{align}
By Lemma~\ref{lem:KL_bd} (see~\eqref{eq:KL_bd}),
\begin{align*}
    \bar\err_\KL^2
    &\le \err_{\chi^2}^2 + \left(2+\max_{0 \le k \le K} \Ren_2(\rho_k^Y\mmid\pi^Y)\right)\,\err_{\chi^2} \,.
\end{align*}

Now, suppose we start with an RGO implementation $\widetilde{\msf R}$ (from Assumption~\ref{ass:sampler}) with the guarantee
\begin{align*}
    \sup_{y\in\R^d} \Ren_3(\widetilde{\msf R}_y \mmid \RGO_y) \le \delta^2 \ll 1\,.
\end{align*}
This implies $\sup_{y\in\R^d} \Ren_2(\widetilde{\msf R}_y \mmid \RGO_y) \le \delta^2 \ll 1$, so
$\sup_{y\in\R^d} \chi^2(\widetilde{\msf R}_y \mmid \RGO_y) \le \exp(\delta^2)-1 \lesssim \delta^2$.
We apply Corollary~\ref{cor:FI_from_chisq} to obtain an RGO implementation $\hRGO$ with guarantees (as defined in~\eqref{Eq:hRGO_acc})
\begin{align*}
    \err_{\chi^2}^2
    &\lesssim \delta^2 + d\sqrt{\delta\,\left(1 + \frac{\log(1/\delta)}{d}\right)}\,, \\[0.25em]
    \bar\err_\FI^2
    &\lesssim
    L\,\left(d+\log(1/\delta)+\max_{0 \le k \le K}\Ren_2(\rho_k^Y \mmid \pi^Y)\right)\,\{d(d+ \log(1/\delta))\,\delta\}^{1/4}\,.
\end{align*}
We simplify these expressions as follows.
First, since $\delta \ll 1$,
\begin{align*}
    \err^2_{\chi^2} \lesssim d\sqrt{\delta\log(1/\delta)} \lesssim d\,.
\end{align*}
Thus, by~\eqref{eq:renyi_along_path},
\begin{align*}
    \max_{0 \le k \le K}\Ren_2(\rho_k^Y \mmid \pi^Y)
    &\lesssim \Delta_0\,\left(1 + \frac{L\log d}{\err^2}\right)\,.
\end{align*}
Substituting this into~\eqref{eq:final_fi_bd},
\begin{align*}
    \FI(\bar\rho_K^X \mmid \pi^X)
    &\lesssim \err^2 + L\,\left(d\sqrt{\delta\log(1/\delta)} + \left(1+\Delta_0\,\left(1 + \frac{L\log d}{\err^2}\right)\right)\,\sqrt{d\sqrt{\delta\log(1/\delta)}} \right) \\
    &\qquad{} + L\,\left(d + \log(1/\delta) + \Delta_0\,\left(1 + \frac{L\log d}{\err^2}\right)\right)\,\{(d+ \log(1/\delta))\,\delta\}^{1/4}\,.
\end{align*}
Finally, we can choose $\delta^{-1} = \poly(\Delta_0, L/\err^2, d)$ to make the error terms at most $\err^2$.
From the proof of Lemma~\ref{Lem:FI_from_chisq_prior}, we see the smoothing variance $t$ that we need to use to construct the RGO implementation $\hRGO$ is $t \asymp h\sqrt{\err_{\chi^2}} \asymp (\delta^2 + d\sqrt{\delta})^{1/4}/L$.
\end{proof}

\subsection{Proof of Theorem~\ref{thm:converse}}
\label{Sec:ThmConverseProof}

\begin{proof}[Proof of Theorem~\ref{thm:converse}]
    The new algorithm $\Alg'$ uses the initialization $\rho_0 = \N(0, L^{-1} I)$, which is known to satisfy $\KL(\rho_0\mmid \pi) \le \KL_0 \deq \frac{d}{2}\log\kappa$ and $\Ren_2(\rho_0 \mmid \pi) \le \Delta_0 \deq \frac{d}{2} \log \kappa$.
    Then, for $k=0,1,2,\dotsc$, $\Alg'$ generates $\rho_{k+1}$ by running $\Alg$ from $\rho_k$ with accuracy parameter $\varepsilon_k^2 = \alpha \KL_0 / 2^k$.
    We prove via induction that $\KL(\rho_k\mmid \pi) \le \KL_0/2^k$, where the base case follows from the preceding discussion.
    Then, by the log-Sobolev inequality,
    \begin{align*}
        \KL(\rho_{k+1} \mmid \pi)
        &\le \frac{1}{2\alpha}\,\FI(\rho_{k+1}\mmid \pi)
        \le \frac{\varepsilon_k^2}{2\alpha}
        \le \frac{\KL_0}{2^{k+1}}\,,
    \end{align*}
    which completes the inductive step.
    We conclude the proof by noting that the $k$-th iteration uses
    \begin{align*}
        \frac{L\, \KL_0/2^k}{\varepsilon_k^2}\,\phi(d)\polylog\Bigl(\Delta_0, \frac{L}{\varepsilon_k^2}\Bigr)
        &= \kappa\,\phi(d) \polylog\Bigl(\Delta_0,\, \frac{\kappa\,2^k}{\KL_0}\Bigr)
        = \kappa\,\phi(d) \polylog(d,\kappa) \poly(k)
    \end{align*}
    queries, and we need at most $K = O(\log\frac{\KL_0}{\varepsilon^2})$ rounds.
    Thus, the total number of queries is at most
    \begin{align*}
        \sum_{k=0}^{K-1}\kappa\,\phi(d) \polylog(d,\kappa) \poly(k)
        &= \kappa \,\phi(d) \polylog(d,\kappa)\poly(K)
        = \kappa \,\phi(d) \polylog\Bigl(d,\kappa, \frac{1}{\varepsilon}\Bigr)\,.
    \end{align*}
    This concludes the proof.
\end{proof}

\newpage
\bibliographystyle{alpha}
\bibliography{refs}

@InProceedings{balasubramanian22a,
  title = 	 {Towards a theory of non-log-concave sampling: first-order stationarity guarantees for {L}angevin {M}onte {C}arlo},
  author =       {Balasubramanian, Krishna and Chewi, Sinho and Erdogdu, Murat A. and Salim, Adil and Zhang, Matthew S.},
  booktitle = 	 {Proceedings of Thirty Fifth Conference on Learning Theory},
  pages = 	 {2896--2923},
  year = 	 {2022},
  editor = 	 {Loh, Po-Ling and Raginsky, Maxim},
  volume = 	 {178},
  series = 	 {Proceedings of Machine Learning Research},
  month = 	 {7},
  publisher =    {PMLR},
}

@inproceedings{wibisono2024score,
  title     = {Optimal score estimation via empirical {B}ayes smoothing},
  author    = {Wibisono, Andre and Wu, Yihong and Yang, Kaylee Yingxi},
  booktitle = {Proceedings of the Conference on Learning Theory},
  year      = {2024},
  series    = {Proceedings of Machine Learning Research},
  publisher = {PMLR}
}

@article{wibisono2025mixing,
  author={Wibisono, Andre},
  journal={IEEE Transactions on Information Theory}, 
  title={{Mixing time of the proximal sampler in relative {F}isher information via strong data processing inequality}}, 
  year={2026},
  volume={72},
  number={5},
  pages={2667--2688},}

@article{Che+25DDPM,
      title={{DDPM} score matching and distribution learning}, 
      author={Sinho Chewi and Alkis Kalavasis and Anay Mehrotra and Omar Montasser},
      year={2025},
      journal={arXiv preprint 2504.05161},
}

@InProceedings{Che+23FisherLower,
  title = 	 {Fisher information lower bounds for sampling},
  author =       {Chewi, Sinho and Gerber, Patrik R. and Lee, Holden and Lu, Chen},
  booktitle = 	 {Proceedings of the 34th International Conference on Algorithmic Learning Theory},
  pages = 	 {375--410},
  year = 	 {2023},
  editor = 	 {Agrawal, Shipra and Orabona, Francesco},
  volume = 	 {201},
  series = 	 {Proceedings of Machine Learning Research},
  month = 	 {2},
  publisher =    {PMLR},
}

@article {AltChe24Warm,
    AUTHOR = {Altschuler, Jason M. and Chewi, Sinho},
     TITLE = {Faster high-accuracy log-concave sampling via algorithmic warm starts},
   JOURNAL = {J. ACM},
  FJOURNAL = {Journal of the ACM},
    VOLUME = {71},
      YEAR = {2024},
    NUMBER = {3},
     PAGES = {Art. 24, 55},
}

@inproceedings{LST21structured,
  author = {Lee, Yin Tat and Shen, Ruoqi and Tian, Kevin},
  booktitle = {Conference on Learning Theory},
  pages = {2993--3050},
  publisher = {PMLR},
  title = {Structured logconcave sampling with a restricted {G}aussian oracle},
  volume = {134},
  year = {2021}}

@inproceedings{CCSW22improved,
  author = {Chen, Yongxin and Chewi, Sinho and Salim, Adil and Wibisono, Andre},
  booktitle = {Conference on Learning Theory},
  pages = {2984--3014},
  publisher = {PMLR},
  title = {Improved analysis for a proximal algorithm for sampling},
  volume = {178},
  year = {2022}}

@book {Chewi26Book,
    AUTHOR = {Chewi, Sinho},
     TITLE = {Log-concave sampling},
 PUBLISHER = {Forthcoming},
      YEAR = {2026},
     NOTE = {Available online at \url{https://chewisinho.github.io/}}
}

@inproceedings{LeeRisGe18Multimodal,
 author = {Lee, Holden and Risteski, Andrej and Ge, Rong},
 booktitle = {Advances in Neural Information Processing Systems},
 editor = {S. Bengio and H. Wallach and H. Larochelle and K. Grauman and N. Cesa-Bianchi and R. Garnett},
 pages = {},
 publisher = {Curran Associates, Inc.},
 title = {Beyond log-concavity: provable guarantees for sampling multi-modal distributions using simulated tempering {L}angevin {M}onte {C}arlo},
 volume = {31},
 year = {2018}
}

@article{Nes12GradSmall,
  title={How to make the gradients small},
  author={Nesterov, Yurii},
  journal={Optima. Mathematical Optimization Society Newsletter},
  number={88},
  pages={10--11},
  year={2012}
}

@article {Car+20StatPt,
    AUTHOR = {Carmon, Yair and Duchi, John C. and Hinder, Oliver and
              Sidford, Aaron},
     TITLE = {Lower bounds for finding stationary points {I}},
   JOURNAL = {Math. Program.},
  FJOURNAL = {Mathematical Programming},
    VOLUME = {184},
      YEAR = {2020},
    NUMBER = {1-2, Ser. A},
     PAGES = {71--120},
}

@article {JKO,
    AUTHOR = {Jordan, Richard and Kinderlehrer, David and Otto, Felix},
     TITLE = {The variational formulation of the {F}okker--{P}lanck equation},
   JOURNAL = {SIAM J. Math. Anal.},
  FJOURNAL = {SIAM Journal on Mathematical Analysis},
    VOLUME = {29},
      YEAR = {1998},
    NUMBER = {1},
     PAGES = {1--17},
}

@InProceedings{Wib18SamplingOpt, 
title = {Sampling as optimization in the space of measures: the {L}angevin dynamics as a composite optimization problem}, 
author = {Wibisono, Andre}, 
booktitle = {Proceedings of the 31st Conference on Learning Theory}, 
pages = {2093--3027}, 
year = {2018}, 
editor = {Sébastien Bubeck and Vianney Perchet and Philippe Rigollet}, 
volume = {75}, 
series = {Proceedings of Machine Learning Research}, 
address = {}, 
month = {7}, 
publisher = {PMLR}, }

@InProceedings{BubMik20GradFlow,
  title = 	 {How to trap a gradient flow},
  author =       {Bubeck, S\'ebastien and Mikulincer, Dan},
  booktitle = 	 {Proceedings of Thirty Third Conference on Learning Theory},
  pages = 	 {940--960},
  year = 	 {2020},
  editor = 	 {Abernethy, Jacob and Agarwal, Shivani},
  volume = 	 {125},
  series = 	 {Proceedings of Machine Learning Research},
  month = 	 {7},
  publisher =    {PMLR},
}

@InProceedings{CheBubSal23Stat,
  title = 	 {On the complexity of finding stationary points of smooth functions in one dimension},
  author =       {Chewi, Sinho and Bubeck, S\'ebastien and Salim, Adil},
  booktitle = 	 {Proceedings of the 34th International Conference on Algorithmic Learning Theory},
  pages = 	 {358--374},
  year = 	 {2023},
  editor = 	 {Agrawal, Shipra and Orabona, Francesco},
  volume = 	 {201},
  series = 	 {Proceedings of Machine Learning Research},
  month = 	 {2},
  publisher =    {PMLR},
}

@InProceedings{HolZam23Stat,
  title = 	 {The computational complexity of finding stationary points in non-convex optimization},
  author =       {Hollender, Alexandros and Zampetakis, Manolis},
  booktitle = 	 {Proceedings of Thirty Sixth Conference on Learning Theory},
  pages = 	 {5571--5572},
  year = 	 {2023},
  editor = 	 {Neu, Gergely and Rosasco, Lorenzo},
  volume = 	 {195},
  series = 	 {Proceedings of Machine Learning Research},
  month = 	 {7},
  publisher =    {PMLR},
}

@InProceedings{FanYuaChe23ImprovedProx,
  title = 	 {Improved dimension dependence of a proximal algorithm for sampling},
  author =       {Fan, Jiaojiao and Yuan, Bo and Chen, Yongxin},
  booktitle = 	 {Proceedings of Thirty Sixth Conference on Learning Theory},
  pages = 	 {1473--1521},
  year = 	 {2023},
  editor = 	 {Neu, Gergely and Rosasco, Lorenzo},
  volume = 	 {195},
  series = 	 {Proceedings of Machine Learning Research},
  month = 	 {7},
  publisher =    {PMLR},
}

@article{Chen+26HighAccDiffusion,
      title={High-accuracy sampling for diffusion models and log-concave distributions}, 
      author={Fan Chen and Sinho Chewi and Constantinos Daskalakis and Alexander Rakhlin},
      year={2026},
      journal={arXiv preprint 2602.01338},
}

@article{Nil26Reweighted,
      title={Reweighted information inequalities}, 
      author={Jonathan Niles-Weed},
      year={2026},
      journal={arXiv preprint 2603.13135},
}

@inproceedings{Cheng+23CondMixing,
 author = {Cheng, Xiang and Wang, Bohan and Zhang, Jingzhao and Zhu, Yusong},
 booktitle = {Advances in Neural Information Processing Systems},
 editor = {A. Oh and T. Naumann and A. Globerson and K. Saenko and M. Hardt and S. Levine},
 pages = {13374--13394},
 publisher = {Curran Associates, Inc.},
 title = {Fast conditional mixing of {MCMC} algorithms for non-log-concave distributions},
 volume = {36},
 year = {2023}
}

@InProceedings{SalSunRic22SVGD,
  title = 	 {A convergence theory for {SVGD} in the population limit under {T}alagrand’s inequality {T1}},
  author =       {Salim, Adil and Sun, Lukang and Richtarik, Peter},
  booktitle = 	 {Proceedings of the 39th International Conference on Machine Learning},
  pages = 	 {19139--19152},
  year = 	 {2022},
  editor = 	 {Chaudhuri, Kamalika and Jegelka, Stefanie and Song, Le and Szepesvari, Csaba and Niu, Gang and Sabato, Sivan},
  volume = 	 {162},
  series = 	 {Proceedings of Machine Learning Research},
  month = 	 {7},
  publisher =    {PMLR},
}

@inproceedings{ShiMac23SVGD,
 author = {Shi, Jiaxin and Mackey, Lester},
 booktitle = {Advances in Neural Information Processing Systems},
 editor = {A. Oh and T. Naumann and A. Globerson and K. Saenko and M. Hardt and S. Levine},
 pages = {26831--26844},
 publisher = {Curran Associates, Inc.},
 title = {A finite-particle convergence rate for {S}tein variational gradient descent},
 volume = {36},
 year = {2023}
}

@InProceedings{SunKarRic23SVGD,
  title = 	 {Convergence of {S}tein variational gradient descent under a weaker smoothness condition},
  author =       {Sun, Lukang and Karagulyan, Avetik and Richtarik, Peter},
  booktitle = 	 {Proceedings of The 26th International Conference on Artificial Intelligence and Statistics},
  pages = 	 {3693--3717},
  year = 	 {2023},
  editor = 	 {Ruiz, Francisco and Dy, Jennifer and van de Meent, Jan-Willem},
  volume = 	 {206},
  series = 	 {Proceedings of Machine Learning Research},
  month = 	 {4},
  publisher =    {PMLR},
}

@article{He+26RSVGD,
      title={Finite-particle rates for regularized {S}tein variational gradient descent}, 
      author={Ye He and Krishnakumar Balasubramanian and Sayan Banerjee and Promit Ghosal},
      year={2026},
      journal={arXiv preprint 2602.05172},
}

@inproceedings{BalBanGho25SVGD,
title={Improved finite-particle convergence rates for {S}tein variational gradient descent},
author={Krishna Balasubramanian and Sayan Banerjee and Promit Ghosal},
booktitle={The Thirteenth International Conference on Learning Representations},
year={2025},
}

@article {He+25RSVGF,
    AUTHOR = {He, Ye and Balasubramanian, Krishnakumar and Sriperumbudur,
              Bharath K. and Lu, Jianfeng},
     TITLE = {Regularized {S}tein variational gradient flow},
   JOURNAL = {Found. Comput. Math.},
  FJOURNAL = {Foundations of Computational Mathematics. The Journal of the
              Society for the Foundations of Computational Mathematics},
    VOLUME = {25},
      YEAR = {2025},
    NUMBER = {4},
     PAGES = {1199--1257},
}

@inproceedings{Liu17SVGD,
 author = {Liu, Qiang},
 booktitle = {Advances in Neural Information Processing Systems},
 editor = {I. Guyon and U. Von Luxburg and S. Bengio and H. Wallach and R. Fergus and S. Vishwanathan and R. Garnett},
 pages = {},
 publisher = {Curran Associates, Inc.},
 title = {Stein variational gradient descent as gradient flow},
 volume = {30},
 year = {2017}
}

@article {Chewi+24QueryLower,
    AUTHOR = {Chewi, Sinho and de Dios Pont, Jaume and Li, Jerry and Lu,
              Chen and Narayanan, Shyam},
     TITLE = {Query lower bounds for log-concave sampling},
   JOURNAL = {J. ACM},
  FJOURNAL = {Journal of the ACM},
    VOLUME = {71},
      YEAR = {2024},
    NUMBER = {4},
     PAGES = {Art. 29, 42},
}

@inproceedings{ZhoSug25FI,
 author = {Zhou, Huanjian and Sugiyama, Masashi},
 booktitle = {Advances in Neural Information Processing Systems},
 editor = {D. Belgrave and C. Zhang and H. Lin and R. Pascanu and P. Koniusz and M. Ghassemi and N. Chen},
 pages = {76789--76812},
 publisher = {Curran Associates, Inc.},
 title = {The adaptive complexity of minimizing relative {F}isher information},
 volume = {38},
 year = {2025}
}

@InProceedings{HeZha25NonLogConcave,
  title = 	 {On the query complexity of sampling from non-log-concave distributions (extended abstract)},
  author =       {He, Yuchen and Zhang, Chihao},
  booktitle = 	 {Proceedings of Thirty Eighth Conference on Learning Theory},
  pages = 	 {2786--2787},
  year = 	 {2025},
  editor = 	 {Haghtalab, Nika and Moitra, Ankur},
  volume = 	 {291},
  series = 	 {Proceedings of Machine Learning Research},
  month = 	 {7},
  publisher =    {PMLR},
}

@inproceedings{GuoTaoChe25Annealed,
 author = {Guo, Wei and Tao, Molei and Chen, Yongxin},
 booktitle = {International Conference on Learning Representations},
 editor = {Y. Yue and A. Garg and N. Peng and F. Sha and R. Yu},
 pages = {78814--78836},
 title = {Provable benefit of annealed {L}angevin {M}onte {C}arlo for non-log-concave sampling},
 volume = {2025},
 year = {2025}
}

\end{document}